\documentclass{article} 
\usepackage{times}
\newif\ifarxiv
\newif\ifcameraready

\camerareadyfalse
\arxivtrue

\ifcameraready
    \usepackage{iclr2027_conference}
    \iclrfinalcopy
\else
    \ifarxiv
        \usepackage{iclr2027_arxiv}
        \iclrfinalcopy
    \else
        \usepackage{iclr2027_conference}
    \fi
\fi

\usepackage{amsmath,amsfonts,bm}

\def\eqref#1{equation~\ref{#1}}

\def\1{\bm{1}}

\DeclareMathAlphabet{\mathsfit}{\encodingdefault}{\sfdefault}{m}{sl}
\SetMathAlphabet{\mathsfit}{bold}{\encodingdefault}{\sfdefault}{bx}{n}

\def\gT{{\mathcal{T}}}

\def\sZ{{\mathbb{Z}}}

\usepackage{hyperref}
\usepackage{url}

\usepackage{amsthm}
\usepackage{algorithm}
\usepackage{algpseudocode}
\usepackage{graphicx}
\usepackage{booktabs}
\usepackage{multirow}
\usepackage[table]{xcolor}
\usepackage{subcaption}
\usepackage{placeins}

\theoremstyle{plain}
\newtheorem{theorem}{Theorem}
\newtheorem{lemma}[theorem]{Lemma}

\newtheorem{corollary}[theorem]{Corollary}

\theoremstyle{definition}
\newtheorem{problem}[theorem]{Problem}

\newcommand{\CCG}{\cellcolor{gray!15}}

\newcommand{\Continue}{\State \textbf{continue}} 

\newcommand{\emptystr}{\lambda}
\newcommand{\ctree}[1]{CT({#1})}
\newcommand{\csuftree}[1]{CST({#1})}
\newcommand{\pdrep}[1]{PD({#1})}
\newcommand{\occ}[2]{\operatorname{occ}({#2},{#1})}
\newcommand{\suffixcol}[1]{\operatorname{psuf}({#1})}
\newcommand{\posstart}[1]{\operatorname{start}({#1})}
\newcommand{\posend}[1]{\operatorname{end}({#1})}
\newcommand{\leafnum}[1]{\operatorname{leaf}({#1})}
\newcommand{\indexof}[1]{\operatorname{idx}({#1})}
\newcommand{\nodestr}[1]{\operatorname{nstr}({#1})}
\newcommand{\realstr}[2]{\operatorname{tstr}({#1})}

\title{CT-Miner: Fast and Coarse-Grained Time-Series Pattern Mining via Cartesian Trees}

\author{Hyundong Jin, Hyunki Hong, Yo-Sub Han\\
Department of Computer Science\\
Yonsei University, Seoul, Republic of Korea \\
\texttt{\{tuzi04,hhong,emmous\}@yonsei.ac.kr}
}

\begin{document}

\maketitle

\begin{abstract}
Time series often contain recurring structural patterns, 
and efficiently mining such patterns into compact representations is essential 
for scalable analysis of long sequences. 
Cartesian tree~(CT) equivalence provides a well-established structural abstraction 
that preserves hierarchical order structure 
while discarding exact values and fine-grained ordinal variations. 
By grouping multiple ordinal patterns into a shared structural form, 
CT equivalence offers a principled way to compress recurring temporal structure.
However, mining frequent CT-equivalent patterns at scale 
remains computationally expensive. 
A naive pairwise approach repeatedly constructs 
and counts CT representations over subsequences, 
requiring $O(n^4)$ time for a sequence of length $n$, 
which severely limits its applicability to long sequences.
We propose a new Cartesian pattern mining algorithm 
based on a Cartesian suffix tree 
that compactly organizes CT-equivalent subsequences 
and reuses shared structural information. 
Our method reduces exhaustive CT-pattern occurrence collection
from $O(n^4)$ to $O(n^2)$ time,
and we formally prove the correctness and complexity bounds.
We further show that this computational gain translates into 
effective compact representations. 
Across diverse time-series datasets, 
a small set of mined CT patterns preserves meaningful clustering structure, 
and comparisons with finer-grained order-preserving representations show that 
CT equivalence reduces redundant ordinal distinctions under limited feature budgets.
Our implementation is available at \url{https://github.com/hyundong98/CT-Miner}.
\end{abstract}

\section{Introduction}
Time-series data arise across a wide range of real-world systems,
making effective representation important for downstream 
analysis~\citep{YueWDYHTX2022, WuHLZWL2023, 
GoswamiSCCLD2024, WangLSYMLJCJ2025, TriratSKNNBKKL2026}.
Computational efficiency and scalability are also important considerations
in modern time-series learning, particularly as sequence lengths and model
scales increase~\citep{FranceschiDJ2019,DempsterSW2021,NieNSK2023,ShiWNLYWJ2025}.
One approach is to summarize long sequences 
through recurring patterns rather than relying directly
on exact numerical values~\citep{ChiuKL2003, ZhouWT2026}.
However, the usefulness of such pattern-based representations 
critically depends on the level of abstraction~\citep{LinKWL2007}. 
Treating small numerical or ordinal variations 
as distinct patterns can lead to highly fragmented 
and redundant representations.

Order-preserving~(OP) patterns address numerical variability 
by representing subsequences through their relative order 
rather than their exact values~\citep{OPP,OPST,EFO,OPF,SOPP}. 
However, OP equivalence preserves the complete pairwise ordering 
among elements, which can be unnecessarily fine-grained 
when the goal is to obtain a compact representation 
of recurring temporal structure. 
Closely related ordinal variations may 
therefore be assigned to different patterns 
even when they share the same higher-level structure. 
This motivates a coarser equivalence relation 
that discards less essential ordering information 
while retaining meaningful structural relationships.

Cartesian tree~(CT) equivalence, 
which has been well studied in theoretical string matching, 
provides such a coarser structural abstraction~\citep{Vuillemin80, ParkALP19, ParkBALP2020, KimH2025}.
Rather than preserving all pairwise order relationships, 
CT equivalence captures the hierarchical ordering structure 
induced by range minima.
Consequently, multiple distinct OP patterns can be represented 
by the same CT pattern, allowing recurring structures 
to be consolidated into a more compact pattern space.
Despite this appealing property,
efficiently extracting frequent CT patterns
from long time series remains challenging.
For a pattern length $L$, direct pairwise counting requires
$O(n^2L)$ time, reaching $O(n^3)$ when $L=\Theta(n)$.
Repeating this computation across all possible lengths
leads to $O(n^4)$ time.
Hash-based aggregation reduces the expected cumulative cost
to $O(nL^2)$, or $O(n^3)$ when $L=\Theta(n)$,
but hash collisions can recover the $O(n^4)$ worst-case bound.
This computational bottleneck limits 
the practical use of CT-based representations at scale.

In this work, we bridge this gap 
with an efficient Cartesian pattern mining algorithm 
based on a Cartesian Suffix Tree. 
Our method exploits shared Cartesian structures across subsequences
by constructing a Cartesian suffix tree in $O(n\log n)$ time.
Once constructed, occurrence information for a fixed pattern length
can be collected in $O(n)$ time, while collecting occurrence
information across all possible pattern lengths takes $O(n^2)$ time.
We formally prove the correctness of the proposed mining procedure
and these complexity bounds.
We further evaluate both the computational efficiency 
and representation utility of the proposed method.  
Compared with a straightforward hash-based CT mining implementation,
our method achieves end-to-end speedups of up to approximately $7.0\times$
as the maximum pattern length increases.
Clustering experiments further show that small CT pattern dictionaries
retain meaningful structural information under limited feature budgets.

\section{Related Work}
Local subsequence structure has also been used directly
for time-series representation and learning.
Shapelet methods, for example, identify discriminative subsequences
for classification~\citep{GrabockaSWS2014}.
Another line of work represents recurring subsequences
through their relative ordering of values.
OPP-Miner represents subsequences by the relative ranks of their values
and formulates the mining of frequent order-preserving patterns~(OPPs),
whose occurrences exceed a minimum support threshold~\citep{OPP}. 
The same work further introduces MOPP-Miner for mining maximal OPPs,
whose order-preserving super-patterns are all infrequent, 
thereby providing a more compact representation of 
the frequent pattern set. 
OPST improves the scalability of maximal OPP mining 
by indexing OP-equivalent substrings 
in an order-preserving suffix tree~\citep{OPST}. 
Once the index is constructed, 
maximal frequent OP patterns can be mined in $O(n)$ time and space,
improving upon the $\Omega(n^3)$ worst-case complexity of MOPP-Miner.
The same framework also supports linear-time mining of 
closed frequent OP patterns.

In contrast to maximal-pattern mining, subsequent approaches largely 
retain threshold-based criteria for selecting patterns of interest.
EFO-Miner improves the efficiency of frequent OPP mining 
by reusing occurrence information from prefix and suffix sub-patterns 
to compute the support of longer patterns~\citep{EFO}. 
It further applies screening and pruning strategies 
to reduce redundant calculations.
Building on EFO-Miner, OPR-Miner further mines 
order-preserving rules between frequent prefix and super-patterns 
using a minimum confidence threshold.
Other studies extend the standard OPP mining formulation 
by modifying the criteria used to identify patterns of interest.
OPF-Miner modifies the support criterion 
by introducing a forgetting mechanism that assigns larger weights 
to more recent occurrences~\citep{OPF}. 
A pattern is considered frequent when 
its forgetting-weighted support exceeds a minimum support threshold.
SOPP-Miner further considers the temporal distribution
of pattern occurrences~\citep{SOPP}. 
It selects patterns whose occurrence gaps satisfy 
a stability threshold based on their coefficient of variation.

Our problem is technically aligned with threshold-based OPP mining, 
in that patterns are selected according to their occurrence frequency. 
However, order-preserving equivalence distinguishes subsequences 
by their complete relative order, 
which can fragment structurally similar trends 
into many separate patterns. 
We instead adopt the coarser Cartesian-tree equivalence 
to consolidate such ordinal variations 
into a more compact structural pattern space, 
and develop an efficient Cartesian suffix tree-based algorithm 
for mining frequent CT patterns.

\section{Preliminary}
An \emph{alphabet} is a set of characters.
A string~$T$ over an alphabet~$\Sigma$
is a sequence of characters from $\Sigma$.
Specifically, if $\Sigma \subseteq \sZ \cup\{\infty,-\infty\}$,
then $T$ is an \emph{integer string}.
We denote the concatenation of
two strings~$T_1$ and $T_2$
by $T_1 \cdot T_2$.
We denote the length of a string~$T$
by $|T|$.
An empty string with zero length
is represented by $\emptystr$.
For a string~$T$,
we denote the $i^{\text{th}}$ character of $T$
by $T[i]$.
A substring of $T$
is the contiguous characters
whose positions are
in the interval $[i:j]$,
and is denoted by $T[i:j]$.
In particular, if $i>j$, then we have $T[i:j]=\emptystr$.
A substring which starts from
the first position of $T$ is a \emph{prefix} of $T$.
Likewise, a substring which ends at
the last position of $T$ is a \emph{suffix} of $T$.

\subsection{Cartesian Tree Equivalence}
The \emph{Cartesian tree}~$\ctree{T}$ of a string~$T$
is a binary tree recursively defined as follows \citep{Vuillemin80}:
\begin{itemize}
    \item If $T=\emptystr$, then $\ctree{T}$ is an empty tree.
    \item Otherwise, let $T[i]$ be the leftmost minimum value among $T$.
    Then $\ctree{T}$ is the tree with $T[i]$ as the root,
    $\ctree{T[1:i-1]}$ as the left subtree,
    and $\ctree{T[i+1:|T|]}$ as the right subtree.
\end{itemize}
Figure~\ref{CTPM:app:fig:ctexample}
in Section~\ref{CTPM:app:ssec:example_figures}
illustrates an example for the Cartesian tree.
The Cartesian tree for any string~$T$
can be built in $O(|T|)$ time \citep{GabowJR1984}.

Two Cartesian trees~$\gT_1$ and $\gT_2$ are equal
if and only if they have the same structure excluding the labels.
We say that two strings~$T_1$ and $T_2$ are \emph{CT equivalent},
or simply \emph{match}, if and only if $\ctree{T_1}=\ctree{T_2}$.
For example, two strings~$T_1=(3,2,7,1,6,4,5,8)$
and $T_2=(8,6,7,1,3,2,4,5)$ are CT equivalent
because $\ctree{T_1}=\ctree{T_2}$.

The \emph{parent-distance (PD) representation} of a string~$T$
is an integer string~$\pdrep{T}$ defined as follows \citep{ParkALP19}:
$$ \pdrep{T}[i]=
\begin{cases}
    i - \max\{j\mid j<i~\text{and}~T[j]\le T[i]\} & \text{if such $j$ exists} \\
    \infty & \text{otherwise}
\end{cases}
$$
\citet{ParkALP19} proved that
$\ctree{T_1}=\ctree{T_2}$ if and only if
$\pdrep{T_1}=\pdrep{T_2}$.
They also proved that
$\pdrep{T}$ for any $T$
can be computed in linear time.
Therefore, the parent-distance representation
works as an alternative representation of
the Cartesian tree of a string.

For a string~$T$,
a string~$p$ \emph{occurs} in $T$
if and only if there exists
a substring $t=T[i:i+|p|)$ of $T$
that is CT equivalent to $p$,
and we say the position~$i$ is an \emph{occurrence} of $p$ in $T$.
The number of occurrences, or \emph{frequency}, of $p$ in $T$
is denoted by $\occ{p}{T}$.

\subsection{Cartesian Suffix Tree}
For a string~$T$, let us denote an ordered set of
the parent-distance representations
of the suffixes of $T$ by
$\suffixcol{T} =\{t_1,t_2,\ldots,t_{|T|+1}\}$,
where $t_i = \pdrep{T[i:|T|]}\cdot(-1)$ for $1\le i \le |T|$,
and $t_{|T|+1}=(-1)$.
Note that the character $-1$
serves as a unique endmarker
for each string.
For each $t_i\in\suffixcol{T}$,
we call the integer~$i$ the \emph{suffix number} of the string.

The \emph{Cartesian suffix tree}~$\csuftree{T}$ of a string~$T$
is a compacted trie
built with the set~$\suffixcol{T}$,
which can be built in worst-case 
$O(|T|\log |T|)$ time
and $O(|T|)$ space \citep{ParkALP19}.
Each edge of the tree
holds the information of
a substring~$t_k[i:j)$
of a string~$t_k\in\suffixcol{T}$.
To reduce the space required to
store this information,
the edges only store
the suffix number~$k$,
the starting position~$i$,
and the ending position~$j$.
In our implementation though,
we store this information
in the node to which the edge points.
We also refer to this information using the name of the node~$n$,
by $\indexof{n}$, $\posstart{n}$, and $\posend{n}$.
Note that $|\suffixcol{T}[\indexof{n}]|=|T|+2-\indexof{n}$.
For a node~$n$,
we denote the number of leaf nodes in the subtree
rooted at $n$ by $\leafnum{n}$.
This piece of information is not present
right after the construction,
but we can preprocess it by a single sweep
of DFS on the tree.

For a node~$n$,
we denote the string obtained
by concatenating strings
corresponding to each edge
in the path from the root to $n$ by
$\nodestr{n}$.
We also define
$\realstr{n}{T} = T[\indexof{n}:\indexof{n}+\posend{n}-1)$;
i.e. the corresponding substring in $T$.
Note that we have $\nodestr{n}=\suffixcol{T}[\indexof{n}][1:\posend{n})$
and $\nodestr{n}=\pdrep{\realstr{n}{T}}$.
One can refer to Figure~\ref{CTPM:app:fig:cstexample}
in Section~\ref{CTPM:app:ssec:example_figures} for an example.

\subsection{Problem Definitions}
The problem that we are tackling is to compute
the frequency of every available pattern.
\begin{problem}[Cartesian tree pattern mining]
\label{CTPM:main:prob:gctpm}
    Given a text~$T$,
    compute the following for
    every string~$p$ with $|p|\le|T|$:
    $$|\{i \mid \ctree{T[i:i+|p|)}=\ctree{p}\}|$$
\end{problem}
A na\"ive solution
would be enumerating through every
substring of $T$,
building their Cartesian tree
(or respective parent-distance representation),
and counting them.
This would take $O(|T|^3)$ time
for each possible length~$|p|\le |T|$,
as there are a total of $O(|T|)$ substrings,
each requiring another $O(|T|)$ time
to construct its Cartesian tree,
and it takes yet another $O(|T|)$ time
on pairwise comparison during counting.
Therefore the na\"ive algorithm
takes $O(|T|^4)$ time overall.
In this paper, we propose an algorithm
that takes $O(|T|^2)$ overall time,
reducing the running time by a factor of $O(|T|^2)$.

\begin{figure}
    \centering
    \includegraphics[page=4, width=0.8\linewidth]{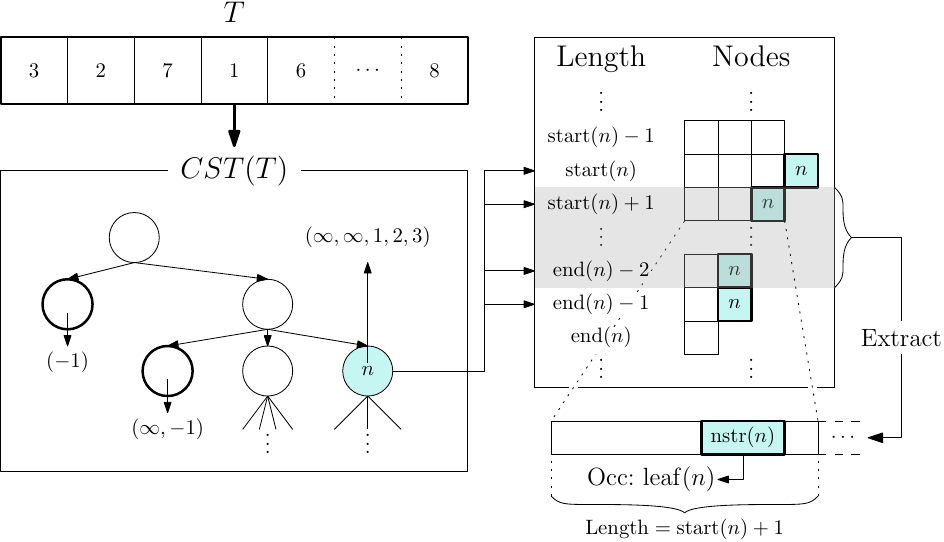}
    \caption{Overview of Algorithm~\ref{CTPM:main:alg:gctpm}.
    It constructs the Cartesian suffix tree of a given string and
    iterates through the nodes to compute the occurrences.}
    \label{CTPM:main:fig:overview}
\end{figure}

\section{Proposed Algorithm}
\subsection{Cartesian Tree Pattern Mining}
In order to solve the problem efficiently,
we need to somehow reduce the bottleneck
caused by brute-forcing through all the substrings of $T$.
The idea here is to utilize a Cartesian suffix tree,
because for a string~$T$,
the number of occurrences of $T[i:j]$ in $T$
is already embedded in some node in $\csuftree{T}$.
The following lemma supports this statement.
\begin{lemma}
\label{CTPM:main:lem:leafcntocc}
    For a non-root node~$n$ in $\csuftree{T}$,
    we have $\occ{\realstr{n}{T}[1:i)}{T}=\leafnum{n}$
    for all $\posstart{n}+1\le i\le\posend{n}$.
\end{lemma}
Therefore, by enumerating through every node in $\csuftree{T}$,
we can obtain the number of occurrences
of each substring of $T$
in constant time,
reducing the time required to search for the occurrences
of each substring.

Algorithm~\ref{CTPM:main:alg:gctpm} implements this approach.
The input for the algorithm is the text~$T$,
and it outputs the frequencies of
every substring of $T$
as an array of lists,
where the $i^{\text{th}}$ list
in the array holds the frequencies
of substrings of $T$ with length $i$.

\begin{algorithm}
\caption{Cartesian tree pattern mining}
\label{CTPM:main:alg:gctpm}
\begin{algorithmic}[1]
\Procedure{CTPM}{$T$}
\State $ST \gets\csuftree{T}$
\State Preprocess all nodes in $ST$ with their leaf count
\label{CTPM:main:ln:a2pre}
\State $N \gets~$List of nodes in $ST$
\label{CTPM:main:ln:a2listnodes}
\State $X \gets$ Array of empty lists, of length $|T|+1$
\For{$i \gets 1,\ldots,|N|$}
\label{CTPM:main:ln:a2mainloop}
    \State $n \gets N[i]$
    \If{($\posend{n} - \posstart{n}=1$ and $n$ is a leaf node) or ($n$ is the root node)}
    \label{CTPM:main:ln:a2supp1}
        \Continue
    \EndIf
    \For{$l\gets\posend{n},\ldots,\posstart{n}+1$}
    \label{CTPM:main:ln:a2reploop}
        \If{$l+\indexof{n}=|T|$}
        \label{CTPM:main:ln:a2supp2}
            \Continue
        \EndIf
        \State Append $n$ to $X[l-1]$
    \EndFor
\EndFor
\State \Return $X$
\EndProcedure
\end{algorithmic}
\end{algorithm}

The algorithm
outputs nodes in $\csuftree{T}$
instead of substrings.
This is to keep the running time
of the algorithm from going above $O(|T|^2)$;
if we explicitly built the substring
for each length~$n$,
then we would take $O(n)$ time for each length,
multiplying the running time by a factor of $O(|T|)$.
By outputting the node instead,
one can always manually retrieve
the parent-distance representation
of the substring
by using the information
embedded in the node, as in the following lemma.

\begin{lemma}
\label{CTPM:main:lem:recoversubstr}
    Given a non-root node~$n$ in a Cartesian suffix tree~$\csuftree{T}$
    of the string~$T$,
    one can obtain $\nodestr{n}$
    and $\occ{\realstr{n}{T}}{T}$ using only
    the information embedded in $n$ and the string~$T$.
\end{lemma}

One can then trim $\nodestr{n}$
to obtain the desired
PD representation of the
CT pattern.
For example, if we had
the node~$n$ in the list~$X[l]$,
then we obtain the corresponding string by
$\nodestr{n}[1:l]$.

The following theorem certifies that
Algorithm~\ref{CTPM:main:alg:gctpm} runs in $O(|T|^2)$ time
and outputs the correct data.
\begin{theorem}
\label{CTPM:main:thm:algsanity2}
    Algorithm~\ref{CTPM:main:alg:gctpm} correctly
    reports frequencies of CT patterns
    for all substrings of length $q\le|T|$
    in $O(n^2)$ time.
\end{theorem}

\subsection{Labelling CT Patterns}
We now try to
identify each element in $X$
with a unique integer
for the sake of further processing.
To do this, we need to build a function
that maps the given parent-distance representation
to a single integer.
For the sake of space efficiency,
we first examine the number of
valid Cartesian trees
of a given length~$n$.

\begin{lemma}
\label{CTPM:main:lem:ctcount}
    Let $C_n=|\{\ctree{T}\mid|T|=n\}|$
    for $n \ge 1$ and $C_0 = 1$.
    Then, for all $n\ge1$, we have
    \begin{equation}
    \label{CTPM:main:eq:ctcount}
        C_{n}=\sum_{i=1}^{n} C_{i-1}C_{n-i},
    \end{equation}
    which coincides with the Catalan numbers.
\end{lemma}

Since Cartesian trees and parent-distance representations
are in a one-to-one correspondence,
the following corollary is immediate from Lemma~\ref{CTPM:main:lem:ctcount}:
\begin{corollary}
\label{CTPM:main:cor:pdcount}
    Let $P_n=|\{\pdrep{T}\mid|T|=n\}|$
    for $n \ge 1$ and $P_0 = 1$.
    We have $P_n=C_n$ for all $n$.
\end{corollary}

Now we are finally ready to build the desired function.
Given a length~$n$,
we can set the range of the label to
$[0:C_n)$ and
map the Cartesian trees (or equivalent PD representations)
to the range.
\begin{theorem}
\label{CTPM:main:thm:labelct}
    Given an integer~$n$,
    one can build a one-to-one function
    $f_n:\{PD(T)\mid |T|=n\}\rightarrow [0:C_n)$
    which maps each valid parent-distance representation
    to a unique integer.
\end{theorem}

Algorithm~\ref{CTPM:main:alg:labelct} specifies the algorithm
for Theorem~\ref{CTPM:main:thm:labelct}.
Given a string~$T$,
the algorithm returns the result of
the function $f_{|T|}(\pdrep{T})$ from Theorem~\ref{CTPM:main:thm:labelct}.
For an explanation of how the function works,
refer to the Section~\ref{CTPM:app:ssec:thmlabelctpf}.

\begin{algorithm}
\caption{Labelling Cartesian tree}
\label{CTPM:main:alg:labelct}
\begin{algorithmic}[1]
\Procedure{LabelCT}{$T$}
\State $n \gets |T|$
\If{$n=0$}
    \State \Return $0$
\EndIf
\State $r \gets \text{Position of rightmost $\infty$ in $\pdrep{T}$}$ 
\State $s\gets\sum_{i=1}^{r-1}C_{i-1}C_{n-i}$     \Comment{Offset from the position of the root}
\State $l_l\gets\Call{LabelCT}{T[1:r-1]}$
\State $l_r\gets\Call{LabelCT}{T[r+1:|T|]}$
\State $p\gets l_l\times C_{n-r}+l_r$     \Comment{Offset from the subtrees}
\State \Return $s+p$
\EndProcedure
\end{algorithmic}
\end{algorithm}

Note that due to how the function works,
two different strings with different lengths
can have the same output from the function.
This can be mitigated by manually making
offsets for different lengths.
For example, if we needed to map
all the Cartesian trees with lengths between $5\le i\le7$,
one can add the offset $C_5$ (resp. $C_5+C_6$) when assigning a label
to the tree with length $6$ (resp. $7$).

\section{Experimental Results and Analysis}
\label{CTPM:main:sec:experimental-results}

We conduct experiments to evaluate the effectiveness
and computational efficiency of the proposed CT mining method.
We evaluate its feature-construction time and memory usage
relative to existing pattern-mining approaches.
We also investigate whether CT patterns provide compact representations
that preserve meaningful clustering structure under limited
feature budgets.
In order to better understand the effect of CT abstraction,
we further analyze how fine-grained OP patterns are consolidated
into shared CT features.
Finally, we examine the computational benefits of our implementation
over naive CT mining as the input length
and maximum pattern length increase.

\subsection{Experimental Setting}
\label{CTPM:main:sec:experimental-setup}

For the aggregate comparisons, we evaluate our method 
on five time-series datasets, 
namely Car, Beef, ElectricDevices, MoteStrain, and PigCVP.
Wafer is analyzed separately 
because its unusually extreme clustering behavior
can dominate macro-averaged results.
We therefore use Wafer as an illustrative case study
to examine how CT equivalence consolidates fine-grained OP patterns,
while reporting its complete results in the appendix.
We combine the official training and test splits of each dataset
for unsupervised clustering.
The baselines include Raw, PCA, OPP, EFO, OPF, OPR, MOPP, OPST, and SOPP.
Raw uses the original observations, PCA uses principal components,
and the pattern-based methods represent each series
using occurrences of selected patterns.
We vary the feature budget over $K\in\{10,20,30,40,50\}$
and use $K=30$ for the primary comparison.
Each representation is evaluated using $k$-means with five random seeds.
Ground-truth labels are used only to determine the number of clusters
and evaluate the resulting assignments.
We report normalized mutual information~(NMI),
homogeneity~($h$), and adjusted Rand index~(ARI).

Feature-construction time and memory usage are measured
over five independent runs.
Feature-construction time includes mining or PCA fitting,
feature generation, and invocation overhead, but excludes clustering.
Memory is measured as the maximum sampled sum of resident set size~(RSS)
across the worker process and its descendants 
during feature construction.
Memory comparisons use the common subset of datasets
with valid measurements for every compared miner in all five repetitions.
Detailed implementation choices, hyperparameters, dataset information,
and evaluation protocols are provided 
in Appendix~\ref{CTPM:app:sec:experimental_details}.

\subsection{Overall Quality and Efficiency}

\begin{table}[th]
\centering
\caption{Clustering quality and resource usage at $K=30$, averaged equally over five datasets excluding Wafer. Time includes mining and feature extraction. Bold and underline indicate the best and second-best values, respectively, excluding Raw.}
\label{CTPM:main:tab:main_results}
\small
\begin{tabular}{lrrrrr}
\toprule[1.2pt]
Method & NMI (\%) $\uparrow$ & $h$ (\%) $\uparrow$ & ARI (\%) $\uparrow$ & Time (s) $\downarrow$ & RSS (MiB) $\downarrow$ \\
\midrule[1.2pt]
Raw  & 33.93 & 32.25 & 18.87 & 0.00 & 181.5 \\
\midrule
PCA  & 34.30 & 32.73 & 19.15 & \textbf{0.07} & \textbf{184.6} \\
OPP  & \underline{37.30} & \underline{35.80} & 19.46 & 15.21 & 652.5 \\
EFO  & \underline{37.30} & \underline{35.80} & 19.46 & 12.18 & 640.6 \\
OPF  & 35.83 & 34.25 & 17.50 & 8.55 & 640.7 \\
OPR  & \underline{37.30} & \underline{35.80} & 19.46 & 12.70 & 645.7 \\
MOPP & 35.76 & 34.09 & 18.08 & 13.63 & 647.9 \\
OPST & 35.94 & 35.07 & \textbf{20.36} & 16.63 & 638.0 \\
SOPP & 29.00 & 27.56 & 10.12 & 8.63 & 637.9 \\
\CCG \textbf{CT}
 & \CCG \textbf{37.75}
 & \CCG \textbf{36.24}
 & \CCG \underline{19.95}
 & \CCG \underline{1.29}
 & \CCG \underline{412.2} \\
\bottomrule[1.2pt]
\end{tabular}
\end{table}

Table~\ref{CTPM:main:tab:main_results} summarizes the overall
quality-efficiency trade-off at the primary feature budget of $K=30$.
CT achieves the highest mean NMI and homogeneity, 
reaching $37.75\%$ and $36.24\%$, respectively.
These scores slightly exceed the strongest frequent-OP baselines,
OPP, EFO, and OPR, which obtain $37.30\%$ NMI and $35.80\%$ homogeneity.
For ARI, CT achieves $19.95\%$, second only to OPST at $20.36\%$.
Thus, the coarser CT equivalence retains clustering quality 
comparable to or better than the finer-grained OP representations 
under the same limited feature budget.

The representation quality is accompanied by a substantial reduction
in feature-construction cost.
CT requires only $1.29$ seconds on average, 
compared with $8.55$--$16.63$ seconds 
for the other pattern-mining methods.
Even relative to the fastest pattern-mining baseline, 
OPF, CT is approximately $6.6\times$ faster.
Its mean peak RSS of $412.2$~MiB is also substantially lower than
those of the OP-based miners, 
which require approximately $638$--$653$~MiB, 
although PCA remains both faster and more memory-efficient.
Taken together, these results show that CT does not obtain 
its computational advantage by sacrificing representation quality.
It provides competitive clustering features while substantially
reducing the cost of mining them.

\subsection{Clustering Quality within Limited Feature Budgets}

\begin{figure}[htbp]
\centering
\includegraphics[width=\linewidth]{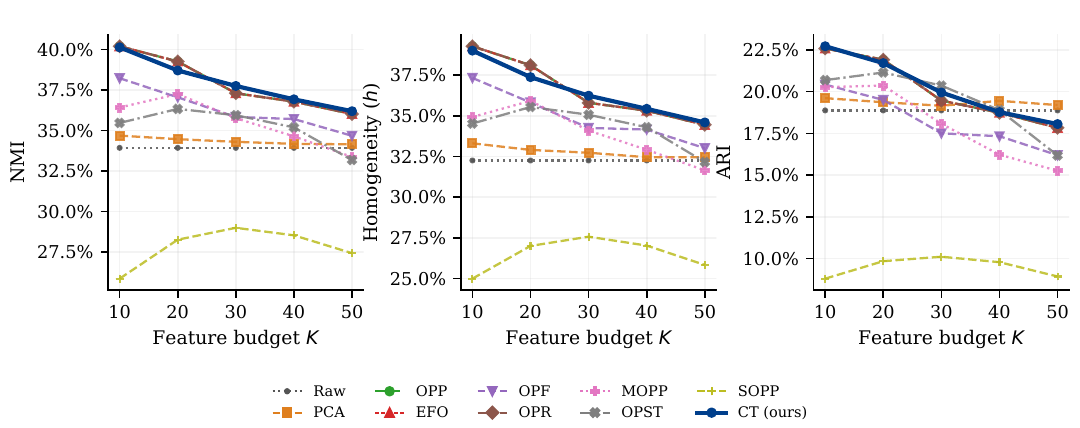}
\caption{Clustering quality versus feature budget, averaged over the five main datasets. Only the evaluated budgets are shown as ticks.}
\label{CTPM:main:fig:budget}
\end{figure}

Figure~\ref{CTPM:main:fig:budget} examines whether the
clustering results depend strongly on the choice of feature budget.
Across $K\in\{10,20,30,40,50\}$, CT remains close to the
strongest OP-based representations on all three clustering metrics.
Its performance remains relatively stable across $K$,
indicating that the result at the primary budget of $K=30$ is not 
an isolated operating point.
A notable trend across the pattern-based methods is that
increasing the number of selected features does not consistently
improve clustering quality.
For several methods, including CT, the best scores occur at
smaller or intermediate budgets and gradually decrease 
as additional patterns are included.
This suggests that feature budget controls not only 
representation capacity 
but also the amount of less informative structure 
introduced into the clustering space.
Accordingly, the quality of the selected patterns is at least 
as important as the size of the dictionary.
Overall, CT remains competitive across the evaluated budgets,
supporting its use as a compact representation.

\subsection{Why Does CT Abstraction Help?}

The preceding results show that CT retains 
competitive clustering quality under limited feature budgets 
while substantially reducing feature-construction cost.
We next examine how the coarser CT equivalence changes the selected
representation at the feature level.
We use Wafer as an illustrative case study.
Wafer is excluded from the five-dataset aggregate results because its
clustering behavior is unusually extreme and can dominate macro averages.
We therefore use it here only to analyze how OP patterns are consolidated
under CT equivalence, rather than as evidence for average performance.
A key effect of CT abstraction is that multiple OP patterns
with different complete relative orders can share 
the same Cartesian-tree structure.
Their occurrences are therefore accumulated into a single CT feature.
Table~\ref{CTPM:main:tab:ct_aggregation} shows representative examples.
For instance, CT pattern $(0,0,1)$ merges
the OP patterns $(2,1,3)$ and $(3,1,2)$,
and its occurrence-count feature is exactly 
the sum of the two corresponding OP features.
This aggregation moves the shared structural pattern 
to rank 8 in the CT dictionary, 
while the two constituent patterns appear 
at ranks 18 and 19 under EFO.

\begin{table}[ht]
\centering
\caption{
Examples of OP patterns consolidated by CT equivalence on Wafer.
Ranks are based on frequency-based feature selection.
$>K$ indicates that the OP pattern is outside the top-$K$ dictionary.
}
\label{CTPM:main:tab:ct_aggregation}
\small
\setlength{\tabcolsep}{4.5pt}
\begin{tabular}{llcc}
\toprule[1.2pt]
CT pattern & Constituent OP patterns & CT rank & EFO rank(s) \\
\midrule[1.2pt]
$(0,0,1)$
& $(2,1,3)$, $(3,1,2)$
& 8 & 18, 19 \\

$(0,1,0,1)$
& $(2,3,1,4)$, $(2,4,1,3)$, $(3,4,1,2)$
& 25 & 31, $>50$, $>50$ \\

$(0,0,1,0)$
& $(3,2,4,1)$, $(4,2,3,1)$
& 29 & 30, $>50$ \\
\bottomrule[1.2pt]
\end{tabular}
\end{table}

The same effect occurs for longer patterns.
CT pattern $(0,1,0,1)$ consolidates three distinct OP patterns,
only one of which appears in EFO's top-50 dictionary,
while $(0,0,1,0)$ combines two OP patterns,
one of which falls outside the top 50.
Thus, CT equivalence can concentrate support that is fragmented
across several finer-grained OP classes.
Under a limited feature budget, this changes how representational
capacity is allocated: one feature can represent a recurring
structural pattern instead of several ordinal variants.

\begin{figure}[ht]
    \centering
    \includegraphics[width=\linewidth]{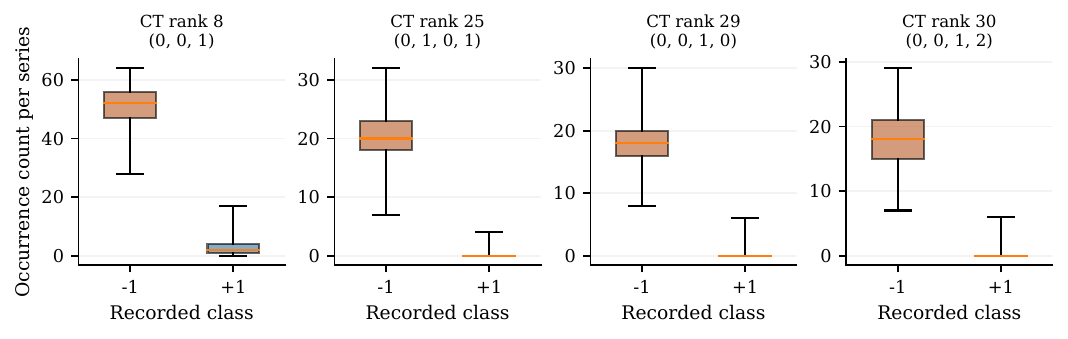}
    \caption{
    Occurrence counts of representative CT features on Wafer,
    grouped by the recorded class.
    Boxes show quartiles and medians, and whiskers span the observed range.
    The analysis is post-hoc and uses labels only for interpretation;
    labels are not used for feature selection or clustering.
    }
    \label{CTPM:main:fig:wafer_features}
\end{figure}

Figure~\ref{CTPM:main:fig:wafer_features} provides a complementary
view of the selected CT features.
Several features that enter the dictionary near $K=30$
show well-separated occurrence-count distributions
between the two recorded classes.
This observation helps explain how a coarser representation
can remain effective under a small feature budget.
Importantly, it does not imply that CT always produces fewer patterns.
Under threshold-based mining, several low-support OP classes
may merge into a CT class whose aggregated support exceeds the threshold,
so CT can in fact yield more qualifying features.
The benefit is instead that CT pools ordinal variations
into shared structural features and can promote those structures
within a frequency-ranked dictionary.
Finally, although several of these Wafer features have
well-separated class-wise occurrence ranges,
we do not attribute the clustering result to any individual feature.
No leave-one-feature-out ablation was performed.
The analysis therefore provides a mechanistic illustration,
rather than a causal explanation, of how CT abstraction
can preserve useful structure while reducing distinctions
between fine-grained ordinal patterns.

\subsection{Time and Memory against a Naive Implementation}

We compare CT with CT-Hash, a naive implementation
that independently constructs and hashes the Cartesian-tree
representation of each window.
Table~\ref{CTPM:main:tab:comparison_naive} compares CT with CT-Hash
on three datasets whose sequence lengths allow the maximum pattern length
to be increased up to $L=256$.
Across all three datasets, CT shows an increasing runtime advantage
as longer patterns are explored.
This trend is consistent with the different dependence on $L$
in the theoretical time bounds.
For fixed input length $n$, CT's $O(nL)$ bound is linear in $L$,
whereas CT-Hash's expected $O(nL^2)$ bound is quadratic.
By reusing structural information across overlapping windows,
CT avoids repeatedly encoding each window as the maximum pattern
length increases.
Matching outputs across all paired runs confirm that this reduction
in computation preserves the mining results.
The runtime improvement is accompanied by a time-space trade-off.
For Car and Beef, CT uses less memory than CT-Hash over most of the
evaluated range, although its memory usage eventually becomes comparable
to or slightly higher than CT-Hash at the largest $L$.
The trade-off is more pronounced on PigCVP, where CT achieves a growing
runtime advantage while its memory footprint increases substantially
with $L$.
These results show that retaining shared Cartesian structures can
considerably reduce repeated computation for longer-pattern mining,
at the cost of additional storage in demanding settings.

\begin{table}[ht]
\centering
\caption{Comparison with the naive CT-Hash implementation over increasing
maximum pattern length $L$. We report mean feature-extraction time over
five runs and peak process RSS.}
\label{CTPM:main:tab:comparison_naive}
\small
\setlength{\tabcolsep}{4.5pt}
\begin{tabular}{lllrrrrrr}
\toprule[1.2pt]
Dataset & Algorithm & Measure & $L=8$ & $L=16$ & $L=32$ & $L=64$ & $L=128$ & $L=256$ \\
\midrule[1.2pt]
\multirow{4}{*}[-3pt]{\shortstack[c]{Car\\($n=577$)}} & \multirow{2}{*}{CT} & Time. (s) & 0.357 & 0.508 & 0.818 & 1.370 & 2.333 & 3.656 \\
 &  & Mem. (MiB) & 127.2 & 152.0 & 251.4 & 490.3 & 692.5 & 901.4 \\
\cmidrule(lr){2-9}
 & \multirow{2}{*}{CT-Hash} & Time. (s) & 0.467 & 0.830 & 1.700 & 3.929 & 10.162 & 25.691 \\
 &  & Mem. (MiB) & 225.0 & 651.3 & 652.7 & 676.7 & 708.8 & 723.4 \\
\midrule
\multirow{4}{*}[-3pt]{\shortstack[c]{Beef\\($n=470$)}} & \multirow{2}{*}{CT} & Time. (s) & 0.283 & 0.353 & 0.563 & 0.974 & 1.456 & 1.819 \\
 &  & Mem. (MiB) & 127.2 & 127.2 & 166.3 & 322.4 & 578.0 & 681.7 \\
\cmidrule(lr){2-9}
 & \multirow{2}{*}{CT-Hash} & Time. (s) & 0.303 & 0.493 & 1.027 & 1.926 & 4.151 & 9.163 \\
 &  & Mem. (MiB) & 132.8 & 321.2 & 640.6 & 658.3 & 661.3 & 650.3 \\
\midrule
\multirow{4}{*}[-3pt]{\shortstack[c]{PigCVP\\($n=2000$)}} & \multirow{2}{*}{CT} & Time. (s) & 1.697 & 2.712 & 4.369 & 6.692 & 10.555 & 17.398 \\
 &  & Mem. (MiB) & 551.1 & 667.8 & 993.9 & 1,436.2 & 3,761.1 & 6,841.0 \\
\cmidrule(lr){2-9}
 & \multirow{2}{*}{CT-Hash} & Time. (s) & 1.669 & 4.766 & 10.206 & 18.916 & 33.675 & 62.629 \\
 &  & Mem. (MiB) & 653.1 & 730.4 & 739.7 & 755.6 & 752.3 & 753.7 \\
\bottomrule[1.2pt]
\end{tabular}
\end{table}

\section{Conclusion}

We introduced CT-Miner, an efficient framework for frequent pattern mining
under Cartesian-tree~(CT) equivalence.
By grouping ordinal variations that share the same hierarchical structure,
CT provides a coarser alternative to order-preserving representations.
Using a Cartesian suffix tree, CT-Miner reuses shared structure across
subsequences and reduces exhaustive occurrence collection to $O(n^2)$ time,
while supporting efficient mining at individual pattern lengths.
Experiments show that CT retains competitive clustering quality under limited
feature budgets while substantially reducing feature-construction cost
relative to existing pattern-mining approaches.
Our feature-level analysis further shows that multiple fine-grained OP patterns
can be consolidated into shared CT features, concentrating recurring structure
that would otherwise be distributed across separate ordinal patterns.
Controlled comparisons with a hash-based implementation also demonstrate
increasing runtime benefits as longer patterns are explored.
Together, these results suggest that CT equivalence provides a useful
structural abstraction for compact and efficient time-series pattern mining.

\subsection*{AI use statement}

In this work, we used generative AI tools to assist with the implementation
of experimental code.
We have not used generative AI tools for other tasks requiring disclosure,
and the rest of the required-disclosure tasks are not applicable to this work.
Additionally, we used generative AI tools to edit the manuscript for
readability and language polishing.
We have reviewed all AI-assisted work.
AI-assisted experimental code was manually reviewed and tested for correctness,
and all AI-assisted manuscript edits were reviewed by the authors.
We take responsibility for the final content of this work,
including text, claims or artifacts produced with the aid of generative AI.

\bibliography{iclr2027_conference}

@inproceedings{ParkALP19,
  author       = {Sung Gwan Park and
                  Amihood Amir and
                  Gad M. Landau and
                  Kunsoo Park},
  title        = {Cartesian Tree Matching and Indexing},
  booktitle    = {Proceedings of the 30th Annual Symposium on Combinatorial Pattern Matching},
  volume       = {128},
  pages        = {16:1--16:14},
  year         = {2019}
}

@article{Vuillemin80,
  author       = {Jean Vuillemin},
  title        = {A Unifying Look at Data Structures},
  journal      = {Communications of the ACM},
  volume       = {23},
  number       = {4},
  pages        = {229--239},
  year         = {1980}
}

@article{OPP,
  title={{OPP-Miner}: Order-preserving sequential pattern mining for time series},
  author={Wu, Youxi and Hu, Qian and Li, Yan and Guo, Lei and Zhu, Xingquan and Wu, Xindong},
  journal={IEEE Transactions on Cybernetics},
  volume={53},
  number={5},
  pages={3288--3300},
  year={2023},
}

@article{EFO,
  title={{OPR-Miner}: Order-preserving rule mining for time series},
  author={Wu, Youxi and Zhao, Xiaoqian and Li, Yan and Guo, Lei and Zhu, Xingquan and Fournier-Viger, Philippe and Wu, Xindong},
  journal={IEEE Transactions on Knowledge and Data Engineering},
  volume={35},
  number={11},
  pages={11722--11735},
  year={2023},
}

@article{OPF,
  title={{OPF-Miner}: Order-preserving pattern mining with forgetting mechanism for time series},
  author={Li, Yan and Ma, Chenyu and Gao, Rong and Wu, Youxi and Li, Jinyan and Wang, Wenjian and Wu, Xindong},
  journal={IEEE Transactions on Knowledge and Data Engineering},
  volume={36},
  number={12},
  pages={8981--8995},
  year={2024},
}

@inproceedings{OPST,
  title={Scalable order-preserving pattern mining},
  author={Li, Ling and Zuba, Wiktor and Loukides, Grigorios and Pissis, Solon P. and Matsangidou, Maria},
  booktitle={2024 IEEE International Conference on Data Mining},
  pages={211--220},
  year={2024},
}

@article{SOPP,
  title={{SOPP-Miner}: Stable Distribution Order-Preserving Pattern Mining for Time Series},
  author={Zhang, Yajie and Wu, Youxi and Li, Yan},
  journal={IEEE Transactions on Knowledge and Data Engineering},
  volume={38},
  number={8},
  pages={5406--5419},
  year={2026},
}

@article{DauBKYZGRK2019,
  title={The {UCR} time series archive},
  author={Dau, Hoang Anh and Bagnall, Anthony and Kamgar, Kaveh and Yeh, Chin-Chia Michael and Zhu, Yan and Gharghabi, Shaghayegh and Ratanamahatana, Chotirat Ann and Keogh, Eamonn},
  journal={IEEE/CAA Journal of Automatica Sinica},
  volume={6},
  number={6},
  pages={1293--1305},
  year={2019},
}

@article{ParkBALP2020,
  title={Finding patterns and periods in Cartesian tree matching},
  author={Park, Sung Gwan and Bataa, Magsarjav and Amir, Amihood and Landau, Gad M. and Park, Kunsoo},
  journal={Theoretical Computer Science},
  volume={845},
  pages={181--197},
  year={2020},
}

@inproceedings{GabowJR1984,
  author       = {Harold N. Gabow and
                  Jon Louis Bentley and
                  Robert Endre Tarjan},
  title        = {Scaling and Related Techniques for Geometry Problems},
  booktitle    = {Proceedings of the 16th Annual {ACM} Symposium on Theory of Computing},
  pages        = {135--143},
  year         = {1984},
}

@article{TriratSKNNBKKL2026,
author = {Trirat, Patara and Shin, Yooju and Kang, Junhyeok and Nam, Youngeun and Na, Jihye and Bae, Minyoung and Kim, Joeun and Kim, Byunghyun and Lee, Jae-Gil},
title = {Universal Time-Series Representation Learning: A Survey},
year = {2026},
volume = {58},
number = {12},
journal = {ACM Computing Surveys},
pages = {1--40},
}

@article{ZhouWT2026,
  title={Order-preserving pattern matching: Review},
  author={Zhou, Yu and Wang, Feng and Tang, Yongning},
  year={2026},
  volume = {38},
  number = {3},
  journal={IEEE Transactions on Knowledge and Data Engineering},
  pages = {1885--1904}
}

@inproceedings{WuHLZWL2023,
  title        = {{TimesNet}: Temporal 2D-Variation Modeling for General Time Series Analysis},
  author       = {Haixu Wu and
                  Tengge Hu and
                  Yong Liu and
                  Hang Zhou and
                  Jianmin Wang and
                  Mingsheng Long},
  booktitle    = {Proceedings of the 11th International Conference on Learning Representations},
  year         = {2023},
}

@inproceedings{ChiuKL2003,
  title={Probabilistic discovery of time series motifs},
  author={Chiu, Bill and Keogh, Eamonn and Lonardi, Stefano},
  booktitle={Proceedings of the 9th {ACM} {SIGKDD} International Conference on Knowledge Discovery and Data Mining},
  pages={493--498},
  year={2003}
}

@article{LinKWL2007,
  title={{Experiencing SAX}: a novel symbolic representation of time series},
  author={Lin, Jessica and Keogh, Eamonn and Wei, Li and Lonardi, Stefano},
  journal={Data Mining and Knowledge Discovery},
  volume={15},
  number={2},
  pages={107--144},
  year={2007},
}

@inproceedings{FranceschiDJ2019,
  author       = {Jean{-}Yves Franceschi and
                  Aymeric Dieuleveut and
                  Martin Jaggi},
  title        = {Unsupervised Scalable Representation Learning for Multivariate Time Series},
  booktitle    = {Advances in Neural Information Processing Systems},
  volume       = {32},
  pages        = {4652--4663},
  year         = {2019},
}

@inproceedings{YueWDYHTX2022,
  author       = {Zhihan Yue and
                  Yujing Wang and
                  Juanyong Duan and
                  Tianmeng Yang and
                  Congrui Huang and
                  Yunhai Tong and
                  Bixiong Xu},
  title        = {{TS2Vec}: Towards Universal Representation of Time Series},
  booktitle    = {Proceedings of the 36th {AAAI} Conference on Artificial Intelligence},
  volume       = {36},
  pages        = {8980--8987},
  year         = {2022},
}

@inproceedings{DempsterSW2021,
  author       = {Angus Dempster and
                  Daniel F. Schmidt and
                  Geoffrey I. Webb},
  title        = {{MiniRocket}: {A} Very Fast (Almost) Deterministic Transform for Time Series Classification},
  booktitle    = {Proceedings of the 27th {ACM} {SIGKDD} Conference on Knowledge Discovery and Data Mining},
  pages        = {248--257},
  year         = {2021},
}

@inproceedings{NieNSK2023,
  title        = {A Time Series is Worth 64 Words: Long-term Forecasting with Transformers},
  author       = {Yuqi Nie and
                  Nam H. Nguyen and
                  Phanwadee Sinthong and
                  Jayant Kalagnanam},
  booktitle    = {Proceedings of the 11th International Conference on Learning Representations},
  year         = {2023},
}

@inproceedings{GoswamiSCCLD2024,
  author       = {Mononito Goswami and
                  Konrad Szafer and
                  Arjun Choudhry and
                  Yifu Cai and
                  Shuo Li and
                  Artur Dubrawski},
  title        = {{MOMENT}: {A} Family of Open Time-series Foundation Models},
  booktitle    = {Proceedings of the 41st International Conference on Machine Learning},
  volume       = {235},
  pages        = {16115--16152},
  year         = {2024},
}

@inproceedings{WangLSYMLJCJ2025,
  author       = {Shiyu Wang and
                  Jiawei Li and
                  Xiaoming Shi and
                  Zhou Ye and
                  Baichuan Mo and
                  Wenze Lin and
                  Shengtong Ju and
                  Zhixuan Chu and
                  Ming Jin},
  title        = {{TimeMixer++}: {A} General Time Series Pattern Machine for Universal Predictive Analysis},
  booktitle    = {Proceedings of the 13th International Conference on Learning Representations},
  year         = {2025},
  pages = {16980--17016},
}

@inproceedings{ShiWNLYWJ2025,
  author       = {Xiaoming Shi and
                  Shiyu Wang and
                  Yuqi Nie and
                  Dianqi Li and
                  Zhou Ye and
                  Qingsong Wen and
                  Ming Jin},
  title        = {{Time-MoE}: Billion-Scale Time Series Foundation Models with Mixture of Experts},
  booktitle    = {Proceedings of the 13th International Conference on Learning Representations},
  year         = {2025},
  pages = {34635--34667},
}

@inproceedings{GrabockaSWS2014,
  author       = {Josif Grabocka and
                  Nicolas Schilling and
                  Martin Wistuba and
                  Lars Schmidt{-}Thieme},
  title        = {Learning Time-Series Shapelets},
  booktitle    = {Proceedings of the 20th {ACM} {SIGKDD} International Conference on Knowledge Discovery and Data Mining},
  pages        = {392--401},
  year         = {2014},
}

@article{KimH2025,
  author       = {Sungmin Kim and
                  Yo{-}Sub Han},
  title        = {Approximate Cartesian tree pattern matching},
  journal      = {Theoretical Computer Science},
  volume       = {1056},
  pages        = {115506},
  year         = {2025},
}
\bibliographystyle{iclr2027_conference}

\clearpage
\appendix

\section{Additional Preliminaries}
\subsection{Examples of Cartesian Tree and Cartesian Suffix Tree}
\label{CTPM:app:ssec:example_figures}
In this section,
we present an example of a Cartesian tree and a Cartesian suffix tree.

Consider a string $T=(3, 2, 7, 1, 6, 4, 5, 8)$.
Since $T[4]=1$ is the smallest, it becomes the root,
and $\ctree{T[1:3]}$ and $\ctree{T[5:8]}$ become
the left and right subtrees, respectively.
Figure~\ref{CTPM:app:fig:ctexample} visualizes the Cartesian tree
of the string~$T$.

\begin{figure}[htbp]
    \centering
    \includegraphics[page=1, width=0.7\linewidth]{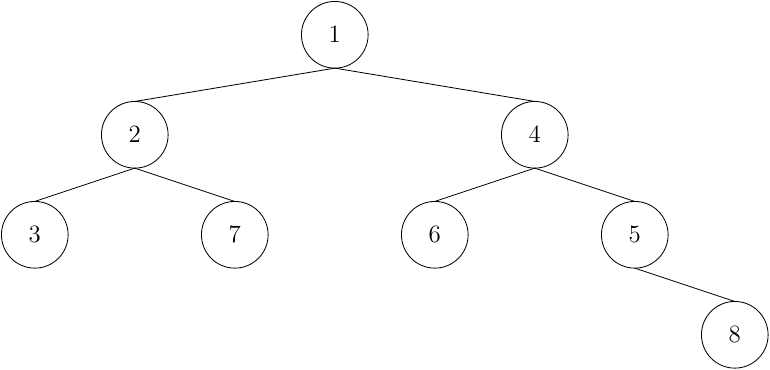}
    \caption{Cartesian tree for $T=(3, 2, 7, 1, 6, 4, 5, 8)$}
    \label{CTPM:app:fig:ctexample}
\end{figure}

For the same string~$T$, its Cartesian suffix tree
is shown in Figure~\ref{CTPM:app:fig:cstexample}.
The tree has $|T|+1=9$ leaf nodes,
where the node~$n_i$ represents
the string with suffix number~$i$;
that is, $\nodestr{n_i}=\pdrep{T[i:|T|]}\cdot(-1)$.
For the node~$m$ in Figure~\ref{CTPM:app:fig:cstexample},
we have
$\posstart{m}=2$
and $\posend{m}=4$.
$\indexof{m}$ can be one of
$1$, $3$, or $5$ depending on the implementation.
All of these values are valid since
$\suffixcol{T}[i][1:4)=(\infty,\infty,1)$
for all $i\in\{1,3,5\}$.

\begin{figure}[htbp]
    \centering
    \includegraphics[page=2, width=0.8\linewidth]{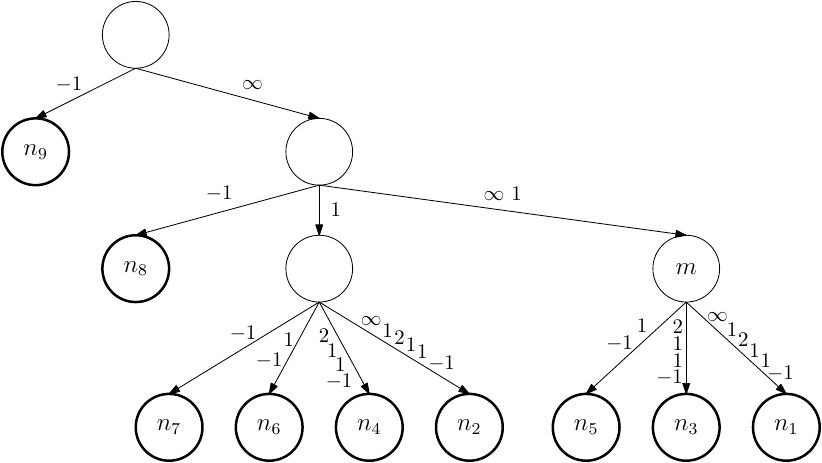}
    \caption{Cartesian suffix tree for $T=(3, 2, 7, 1, 6, 4, 5, 8)$}
    \label{CTPM:app:fig:cstexample}
\end{figure}

\section{Proofs of the Lemmas and Theorems}
\subsection{Proof for Lemma~\ref{CTPM:main:lem:leafcntocc}}
\begin{proof}
    For a non-root node~$n$,
    consider a leaf node~$m$
    in the subtree rooted at $n$.
    Note that $n=m$ if $n$ is already a leaf node.
    Since $n$ is an ancestor of $m$,
    $\nodestr{n}$ is a prefix of $\nodestr{m}$.
    That is, we have
    $$\suffixcol{T}[\indexof{n}][1:\posend{n})=\suffixcol{T}[\indexof{m}][1:\posend{n}).$$
    In other words, we have
    \begin{flalign*}
        \pdrep{\realstr{n}{T}}
        &=\pdrep{T[\indexof{n}:\indexof{n}+\posend{n}-1)}\\
        &=\pdrep{T[\indexof{m}:\indexof{m}+\posend{n}-1)},
    \end{flalign*}
    which implies that the position~$\indexof{m}$ is
    an occurrence of $\realstr{n}{T}$ in $T$.
    There are $\leafnum{n}$ such distinct nodes $m$,
    which proves that $\occ{\realstr{n}{T}[1:i)}{T}=\leafnum{n}$
    for $i=\posend{n}$.
    This argument can be applied to
    all $\posstart{n}+1\le i < \posend{n}$,
    since $\csuftree{T}$ is a compacted trie.
\end{proof}

\subsection{Proof for Lemma~\ref{CTPM:main:lem:recoversubstr}}
\begin{proof}
    By definition, we have
    $$
        \nodestr{n}=\suffixcol{T}[\indexof{n}][1:\posend{n})
        =\pdrep{T[\indexof{n}:|T|]}\cdot(-1)[1:\posend{n}).
    $$
    We can also directly obtain
    the number of occurrences of the string~$\realstr{n}{T}$
    by $\leafnum{n}$ (see Lemma~\ref{CTPM:main:lem:leafcntocc}).
\end{proof}

\subsection{Proof for Theorem~\ref{CTPM:main:thm:algsanity2}}
\begin{proof}
    We first prove the correctness.
    Let us denote the root node by $r$.
    Consider a leaf node~$n$ of $\csuftree{T}$,
    and let $n_0=r,n_1,\cdots,n_{k-1},n_k=n$ be the path
    from the root to the leaf~$n$.
    By Lemma~\ref{CTPM:main:lem:leafcntocc},
    we correctly report all frequencies of
    $T[\indexof{n}:\indexof{n}+i)$
    for all $1\le i <|\suffixcol{T}[\indexof{n}]|$,
    in the loop starting from line~\ref{CTPM:main:ln:a2reploop}.
    Note that we suppress the substring from being reported
    when $i=|\suffixcol{T}[\indexof{n}]|$
    by lines~\ref{CTPM:main:ln:a2supp1} and \ref{CTPM:main:ln:a2supp2},
    because such $i$ is
    out of bounds for $\realstr{n}{T}[1:i)$,
    which is caused by the arbitrary endmarker
    appended to all strings in $\suffixcol{T}$.
    This covers every substring $T[\indexof{n}:i]$
    for all $\indexof{n}\le i\le|T|$.
    Since there exists
    a leaf node~$n_i$ such that
    $\indexof{n_i}=i$ for all $1\le i\le |T|$,
    we can conclude that the algorithm covers
    every single substring
    of $T$.

    We now prove the running time.
    Preprocessing in line~\ref{CTPM:main:ln:a2pre}
    runs in $O(|T|\log |T|)$ time.
    Since the number of nodes in the tree
    is $O(|T|)$,
    obtaining the list of all nodes in line~\ref{CTPM:main:ln:a2listnodes}
    can be done in $O(|T|)$ time,
    by running DFS at the root of the tree.
    The main loop starting from line~\ref{CTPM:main:ln:a2mainloop}
    traverses all nodes in the tree,
    going through all the CT patterns in $T$
    using only $O(1)$ time for each pattern.
    There are $O(|T|^2)$ patterns,
    making the running time $O(|T|^2)$ as well.
\end{proof}

\subsection{Proof for Lemma~\ref{CTPM:main:lem:ctcount}}
\begin{proof}
    We prove this by induction.
    For base case $n=1$, we have $C_1=1$,
    as the only Cartesian tree of length $1$
    is the singleton node.
    This shows that
    the equation~(\ref{CTPM:main:eq:ctcount}) is true for the base case.

    For the inductive case,
    assume that the equation~(\ref{CTPM:main:eq:ctcount})
    holds for all $n\le k$ for some $k$.
    Now suppose we are building a Cartesian tree of
    length~$k+1$,
    whose root is at position~$1\le i\le k+1$.
    Then, since the tree is built recursively from here,
    the left subtree has a total of $C_{i-1}$ possible
    variants, and the right subtree has the total of $C_{k+1-i}$
    variants, resulting in a total of $C_{i-1}C_{k+1-i}$ possible
    trees in this case.
    This is true for all positions~$1\le i\le k+1$,
    so the total number of valid Cartesian trees is
    $$
        \sum_{i=1}^{k+1} C_{i-1}C_{k+1-i},
    $$
    which is equal to $C_{k+1}$.
    This proves the inductive case, completing the proof.
\end{proof}

\subsection{Proof for Theorem~\ref{CTPM:main:thm:labelct}}
\label{CTPM:app:ssec:thmlabelctpf}
\begin{proof}
    The idea here is to
    first split the set of trees
    by the position of the root,
    and then map each of the trees
    to each designated split area.
    We have a base case of
    $f_0(\emptystr)=0$.
    For any other value of $n$,
    given an input~$\pdrep{T}$, we first find the 
    position of the rightmost
    $\infty$ in $\pdrep{T}$,
    which corresponds to the leftmost smallest position in $T$.
    Let the position be $r$.
    Now we calculate the \emph{offset},
    which is the number of Cartesian trees
    whose root position is smaller than $r$.
    This can be computed by
    $$
        s = \sum_{i=1}^{r-1} C_{i-1}C_{n-i}.
    $$
    This value ensures that
    Cartesian trees of root position~$r$
    are mapped to their unique `area'
    of size $C_{r-1}C_{n-r}$.
    Finally, we recursively compute
    $$
        p = f_{r-1}(\pdrep{T[1:r-1]}) \cdot C_{n-r} + f_{n-r}(\pdrep{T[r+1:n]}),
    $$
    which is the position in the `area' of the tree.
    Note that $0\le p < C_{r-1}C_{n-r}$.
    We can then map $PD(T)$ to the value~$s+p$,
    which completes the function.
\end{proof}

\begin{figure}[htbp]
    \centering
    \includegraphics[page=3, width=0.85\linewidth]{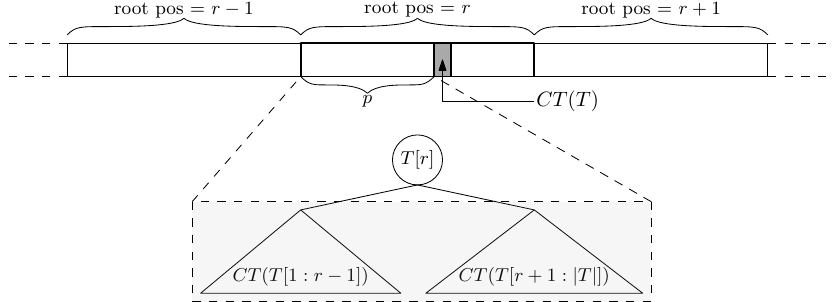}
    \caption{An illustration of Theorem~\ref{CTPM:main:thm:labelct}}
\end{figure}

\section{Experimental Details}
\label{CTPM:app:sec:experimental_details}

For each experimental condition, we average the scores over
five clustering seeds within each dataset.
Aggregate results are then computed by averaging these
dataset-level scores with equal weight,
so that larger datasets do not dominate the reported means.

\subsection{Datasets and Metrics}
\label{CTPM:app:ssec:datasets_metrics}

We evaluate our method on six univariate time-series datasets,
Car, Beef, ElectricDevices, MoteStrain, PigCVP, and Wafer~\citep{DauBKYZGRK2019}.
Our evaluation includes datasets used in previous studies
of order-preserving pattern mining.
Car and Beef were used in the clustering experiments
of \citet{EFO}, while Wafer was used by \citet{OPST}.
For Car and Beef, we expand the evaluation from the
15 series per dataset considered by \citet{EFO}
to the complete collections obtained by combining
the official training and test splits.
We further include ElectricDevices, MoteStrain, and PigCVP
to broaden the evaluation across collection sizes
and sequence lengths.
In particular, ElectricDevices contains 16,637 series,
while PigCVP contains sequences of length 2,000.
Together, the six datasets allow us to assess clustering
quality and computational cost across different data scales.
Table~\ref{CTPM:app:tab:datasets} summarizes their characteristics.

\begin{table}[tbp]
    \centering
    \small
    \caption{Dataset characteristics after combining the official
    training and test splits.}
    \label{CTPM:app:tab:datasets}
    \begin{tabular}{lrrrl}
        \toprule[1.2pt]
        Dataset & Sequences & Length & Classes \\
        \midrule[1.2pt]
        Car & 120 & 577 & 4 \\
        Beef & 60 & 470 & 5 \\
        ElectricDevices & 16,637 & 96 & 7 \\
        MoteStrain & 1,272 & 84 & 2 \\
        PigCVP & 312 & 2,000 & 52 \\
        Wafer & 7,164 & 152 & 2 \\
        \bottomrule[1.2pt]
    \end{tabular}
\end{table}

For each dataset, we combine the official training and test splits
into a single collection for unsupervised clustering.
The representations are constructed and evaluated on this collection.
Ground-truth labels are used only to determine the number of clusters
and evaluate the resulting assignments.
They are not used for pattern selection or representation fitting.
This protocol evaluates in-sample clustering utility.

We assess clustering quality using normalized mutual information~(NMI),
homogeneity~($h$), and the adjusted Rand index~(ARI).
Let $Y$ denote the ground-truth class labels and $Z$ denote
the predicted cluster assignments for the $S$ time series.

Following the definitions used by \citet{SOPP},
we compute NMI with geometric normalization as
\begin{equation}
    \operatorname{NMI}(Y,Z)
    =
    \frac{I(Y,Z)}{\sqrt{H(Y)H(Z)}},
\end{equation}
where $I(Y,Z)$ is the mutual information between the two assignments,
and $H(Y)$ and $H(Z)$ are their respective entropies.
NMI measures the information shared by the predicted clusters
and the ground-truth classes.

Homogeneity measures the extent to which each cluster contains
series from a single ground-truth class and is defined as
\begin{equation}
    h(Y,Z)
    =
    1-\frac{H(Y\mid Z)}{H(Y)},
\end{equation}
where $H(Y\mid Z)$ is the conditional entropy of the class labels
given the cluster assignments.
A homogeneity score of one indicates that every cluster
contains series from only one class.

We additionally report ARI to evaluate pairwise agreement
while adjusting for chance.
Let $a$ denote the number of pairs assigned to the same group
in both $Y$ and $Z$, and let $b$ denote the number of pairs
assigned to different groups in both assignments.
The Rand index is $\operatorname{RI}(Y,Z)=(a+b)/\binom{S}{2}$,
and its adjusted form is
\begin{equation}
    \operatorname{ARI}(Y,Z)
    =
    \frac{
        \operatorname{RI}(Y,Z)
        -\mathbb{E}[\operatorname{RI}(Y,Z)]
    }{
        1-\mathbb{E}[\operatorname{RI}(Y,Z)]
    }.
\end{equation}
The expectation is taken over random assignments
with the class and cluster sizes held fixed.
ARI equals one for identical partitions and has an expected
value of zero under this random-assignment model.
Negative values indicate agreement below chance.
Higher values indicate better clustering quality
for all three metrics.

\subsection{Baseline Implementation Details}
\label{CTPM:app:ssec:baseline-implementations}

We compare the proposed method with Raw, PCA, OPP, EFO, OPR, 
OPF, SOPP, MOPP, and OPST.
Raw and PCA construct features directly from numerical observations.
The remaining baselines represent each series through recurring
relative-order patterns in contiguous subsequences.
The methods differ in how they discover patterns
and which patterns they retain.
We implement all pattern miners within a common Java framework,
using existing Java implementations or ports of released code
as described below.

\paragraph{Raw and PCA.}
Raw uses the observations at individual time positions as features.
A series of length $n$ therefore produces an $n$-dimensional vector.
PCA constructs a lower-dimensional representation by finding
orthogonal directions that capture the greatest variance
across the standardized input series.
Each feature is the projection of a series onto one
of the retained directions.
We compute these directions using randomized
singular value decomposition.

\paragraph{OPP and EFO.}
OPP-Miner~\citep{OPP} and EFO-Miner~\citep{EFO}
retain relative-order patterns that meet a minimum support threshold.
Both generate longer candidates by combining shorter patterns
with compatible overlapping orders.
OPP filters candidate occurrences using adjacent increases
and decreases, then verifies their full relative order.
EFO combines occurrence-position lists of shorter patterns
and prunes candidates that cannot meet the support threshold.
Each selected pattern defines an occurrence-count feature.
The original implementations are in C++.
We use the Java reimplementations released
with OPF-Miner~\citep{OPF}.

\paragraph{OPR.}
OPR-Miner~\citep{EFO} selects feature patterns 
through rules $p \rightarrow q$, 
where $q$ extends $p$ one position to the right.
A rule qualifies when $q$ meets the support threshold and 
its support divided by that of $p$ meets the confidence threshold.
The distinct patterns appearing on either side of qualifying rules
form the feature dictionary.
Each feature records a pattern's occurrence count in a series.
Our Java implementation combines EFO occurrence counting
with the rule-selection procedure from the released C++ code.\footnote{
\url{https://github.com/wuc567/Pattern-Mining/tree/master/OPR-Miner}}

\paragraph{OPF.}
OPF-Miner~\citep{OPF} weights relative-order pattern occurrences
more heavily near the end of each series.
An occurrence ending at position $e$ in a series of length $n$
receives weight $\exp(-\alpha(n-e)/n)$,
where $\alpha$ controls the decay rate.
These weights are summed within each series
to determine both mining support and feature values.
Patterns are retained when their weighted support
meets the minimum threshold.
We adapt the released Java implementation,\footnote{
\url{https://github.com/wuc567/Pattern-Mining/tree/master/OPF-Miner}}
retaining its candidate-generation and support-counting procedures.

\paragraph{SOPP.}
SOPP-Miner~\citep{SOPP} selects feature patterns
based on both occurrence frequency and spacing.
Within each series, it computes the gaps between
consecutive occurrences of a pattern.
The standard deviation of these gaps divided by their mean
measures how much the spacing varies.
A pattern qualifies when its occurrence count meets
the support threshold and this variation
does not exceed the stability threshold.
Patterns are ranked by the number of series
in which they qualify.
Once selected, each pattern defines a feature
whose value is its occurrence count in each series.
We implement a Java port of the released Python code\footnote{
\url{https://github.com/wuc567/Pattern-Mining/tree/master/SOPP-Miner}}.

\paragraph{MOPP.}
MOPP-Miner~\citep{OPP} selects frequent patterns
that have no frequent extension.
It builds longer candidates from shorter patterns
and checks their occurrence counts.
When a longer pattern meets the support threshold,
its immediate prefix and suffix patterns
are excluded from the output because they have
a frequent extension.
The remaining patterns form the feature dictionary,
and each feature records a pattern's occurrence count
in a series.
Our Java port follows this selection procedure
from the released C++ implementation
and uses OPP for occurrence detection.\footnote{
\url{https://github.com/wuc567/Pattern-Mining/tree/master/OPP-Miner}}

\paragraph{OPST.}
OPST~\citep{OPST} uses a suffix-tree index
to organize subsequences by their relative order
and obtain their occurrence counts.
Subsequences with the same relative-order pattern
are represented together in the index,
allowing their counts to be computed from shared structure.
The mining procedure selects frequent patterns
for which no extension to the left or right
meets the support threshold.
These patterns form the feature dictionary,
with occurrence counts as feature values.
We port the released C++ index-construction
and maximal-pattern mining procedures to Java.\footnote{
\url{https://github.com/LINGLI97/OPST}}
Its order representation distinguishes equal values
from strictly ordered values,
allowing patterns to include repeated values.

\paragraph{Adaptation to time-series collections.}
For each method, we collect the distinct patterns
mined from individual series and use them as feature dimensions.
Each series is then mapped to a vector
whose entries record the occurrences of these patterns.
OPF uses forgetting-weighted counts,
while the other pattern-mining methods use ordinary counts.
This mapping gives each feature dimension
a consistent meaning across series within the same method.
For experiments with a fixed feature budget,
$K$ patterns are selected after mining
according to the method-specific ranking criterion.

\subsection{Comparison with Naive Algorithm}
We use CT-Hash to evaluate the benefit of reusing
structural information during CT mining.
CT-Hash is implemented in Java and processes
each contiguous subsequence independently.
For each subsequence, it uses a monotone stack
to construct a code representing its Cartesian-tree structure.
Identical codes are grouped in a hash table
to count occurrences of the corresponding pattern.
Exact code comparisons ensure that hash collisions
do not merge different patterns.
Unlike CT, which reuses its trie index across pattern lengths,
CT-Hash constructs these codes independently
for every subsequence at each evaluated length.
Both implementations use the same pattern-selection criteria
and produce occurrence-count features.
We verify that their selected patterns,
support values, and feature matrices agree for each condition.

\subsection{Hyperparameter Settings}
\label{CTPM:app:ssec:hyperparameters}

For the feature-budget comparison, we use
$K\in\{10,20,30,40,50\}$.
Each pattern-mining method returns its top-$K$ qualifying patterns,
while PCA retains $K$ components.
Raw uses all original observations.
We set the minimum support to $2$,
the OPR confidence threshold to $0.5$,
and the SOPP maximum gap coefficient of variation to $1.0$.
The confidence and stability thresholds are fixed
at the middle values of their respective parameter grids
before evaluating clustering quality.

Support thresholds are applied separately within each series.
For OPF, support is the sum of occurrence weights,
with the decay parameter fixed at $\alpha=1$.
For the other miners, support is the occurrence count.
SOPP requires at least one gap between consecutive occurrences
to evaluate stability.
All miners may consider pattern lengths up to the input length,
subject to their native pruning and termination rules.
CT explores pattern lengths incrementally.
In the fixed-budget experiments, CT safely terminates
when no deeper pattern can enter the current top-$K$ dictionary,
using the fact that support cannot increase under extension.
It also terminates when no pattern at the current length
meets the support threshold.
Thus, the feature budget determines the final dictionary size
and can also reduce the maximum depth explored by CT
through this safe termination criterion.

To examine parameter sensitivity, we vary
the minimum support over $\{5,10,25\}$.
For OPR, each support threshold is combined with
confidence thresholds in $\{0.25,0.5,0.75\}$.
For SOPP, each support threshold is combined with
maximum gap coefficients of variation in $\{0.5,1.0,1.5\}$.
In these experiments, all qualifying patterns are used
as features without a fixed feature budget.
Raw and PCA with 30 components serve as reference representations.
Results are reported separately for each parameter setting.

For the input-length comparison, we retain the first
$20\%,40\%,60\%,80\%$, and $100\%$ of each series
while keeping the number of series fixed.
For an original series of length $n$,
the retained length is $\lfloor pn\rfloor$,
where $p\in\{0.2,0.4,0.6,0.8,1.0\}$.
We use $K=30$ and retain the support, confidence,
stability, and decay settings of the feature-budget comparison.
Mining considers pattern lengths up to the retained input length.
For clustering, we use $k$-means with Lloyd's algorithm
and ten initializations per run.
The number of clusters is set to the number
of ground-truth classes in each dataset.
PCA uses randomized singular value decomposition
with random seed zero and no whitening.

The separate comparison with naive implementation uses $K=30$.
At the original input length, we vary the maximum pattern length
over $L\in\{8,16,32,64,128,256\}$,
excluding settings where $L>n$.
In the input-length sweep, we fix $L=16$
and use the same five prefix proportions.
Both implementations process every length from $2$ to $L$
with early stopping disabled.

\subsection{Evaluation Protocols}
\label{CTPM:app:ssec:evaluation-protocols}

We combine the official training and test splits
of each dataset for unsupervised clustering.
Feature construction and clustering use the same collection.
Ground-truth labels are used only to determine
the number of clusters and evaluate cluster assignments.
All methods receive the same input values,
rounded to nine significant digits,
without extreme-point extraction.
Pattern occurrences are confined to individual series.
OPP, EFO, OPF, OPR, MOPP, SOPP, and CT exclude
windows containing repeated values,
while OPST retains its equality-aware order representation.

Each feature-budget, threshold, and input-length condition
is evaluated independently.
Features are reconstructed for each condition,
including each input prefix.
Before clustering, each column of the Raw and pattern-based
feature matrices is standardized across samples.
PCA is fitted to standardized raw inputs,
and its component scores are used without further standardization.
For each condition, we evaluate the constructed feature matrix
using five $k$-means seeds $\{0,1,2,3,4\}$.
We average NMI, homogeneity, and ARI over these seeds
within each dataset.
Aggregate scores are computed by averaging
the dataset-level scores with equal weight.

We measure feature-construction time and peak memory usage
over five independent runs and report their means.
Feature-construction time includes mining or PCA fitting,
feature generation, and invocation overhead,
but excludes clustering.
Memory is measured as the peak sampled RSS
of the worker process and its descendants
during feature construction.
CT and CT-Hash are evaluated under identical mining settings.
We verify that both produce the same selected patterns,
support values, and feature matrices.
Both methods are executed serially.
Speedup is the ratio of mean CT-Hash mining time
to mean CT mining time, excluding JVM startup
and Python communication.

\subsection{Experimental Environments}
\label{CTPM:app:ssec:environment}

Experiments were conducted on a server running Rocky Linux 8.10
with an AMD Ryzen Threadripper 3960X processor
containing 24 physical cores and 48 hardware threads,
and approximately 256 GiB of system memory.
We used Python 3.11.16 with NumPy 2.1.3,
SciPy 1.15.3, and scikit-learn 1.6.1.
We used aeon~1.3.0 to load the time-series datasets.
Java implementations were compiled and executed
using OpenJDK 25.0.2-internal.
Python numerical libraries were limited to one thread,
and each JVM was configured with one active processor
and a maximum heap size of 16\,GiB.

\section{Additional Experimental Results}
\label{CTPM:app:sec:additional_experimental_results}

\subsection{Full Results on Quality and Resources with Increasing Features}

Tables~\ref{CTPM:app:tab:budget_quality} and~\ref{CTPM:app:tab:budget_resources}
provide the complete dataset-level results across feature budgets.
Unlike the aggregate results in the main text, these tables expose
how the effect of $K$ varies across datasets and evaluation metrics.
Quality values are averaged over five clustering seeds, while
resource measurements are averaged over five independent constructions.
Raw is included as a fixed reference and excluded from the rankings.
Bold and underlined values denote the best and second-best results,
respectively, and CT cells are shaded gray.

The full results reveal that the effect of feature budget is strongly
dataset dependent.
On Car, CT produces the same clustering scores as OPP, EFO, and OPR
across all evaluated budgets, while on Beef their results remain close
with OPF performing better at several intermediate budgets.
The differences are more pronounced on the other datasets.
On MoteStrain, CT remains among the strongest methods across most budgets,
whereas on PigCVP its quality improves substantially from $K=10$ to $K=30$
and remains close to the strongest methods at larger budgets.
ElectricDevices exhibits a different pattern, with CT performing strongly
at small $K$ but gradually decreasing as additional features are included.
OPST also achieves substantially higher ARI than the frequency-based
representations on this dataset despite more moderate NMI differences.
These dataset-level results clarify the non-monotonic trend observed
in the aggregate analysis.
Increasing $K$ can improve, preserve, or reduce clustering quality
depending on the dataset and representation.
The benefit of a compact dictionary therefore does not arise from a
single preferred budget, but from selecting structural features that
remain useful under limited representation capacity.

The resource results show a more consistent pattern for runtime.
CT is the fastest pattern-mining method for every dataset and every
evaluated feature budget, and its advantage remains substantial as
$K$ increases.
In contrast, the memory behavior is more dataset dependent.
CT maintains a much smaller footprint than the OP-based miners on
Car and Beef, and remains lower on MoteStrain despite increasing with $K$.
On PigCVP, its memory usage grows from well below the OP-based methods
at small budgets to a comparable or larger footprint at larger budgets.
ElectricDevices shows the clearest exception, where CT uses more memory
than the other pattern miners throughout the evaluated range.
These results separate the two computational effects of the proposed method.
Its runtime advantage is consistent across both datasets and feature budgets,
whereas its memory advantage depends on the amount of shared Cartesian
structure retained during mining.
The latter motivates the explicit time-space analysis against the naive
CT implementation reported separately.

\begin{table}[th]\centering
\caption{
Clustering quality across feature budgets on the five main datasets.
Values are mean NMI, homogeneity, and ARI over five clustering seeds.
Bold and underline denote the best and second-best values excluding Raw.
}
\label{CTPM:app:tab:budget_quality}
\small
\setlength{\tabcolsep}{2.4pt}
\begin{tabular}{llrrrrrrrrrrrrrrr}
\toprule[1.2pt]
\multirow{2}{*}[-2pt]{Dataset} & \multirow{2}{*}[-2pt]{Method} & \multicolumn{3}{c}{$K=10$} & \multicolumn{3}{c}{$K=20$} & \multicolumn{3}{c}{$K=30$} & \multicolumn{3}{c}{$K=40$} & \multicolumn{3}{c}{$K=50$} \\
\cmidrule(lr){3-5}\cmidrule(lr){6-8}\cmidrule(lr){9-11}\cmidrule(lr){12-14}\cmidrule(lr){15-17}
 &  & NMI & $h$ & ARI & NMI & $h$ & ARI & NMI & $h$ & ARI & NMI & $h$ & ARI & NMI & $h$ & ARI \\
\midrule[1.2pt]
\multirow{10}{*}{Car} & Raw & 22.8 & 21.2 & 12.6 & 22.8 & 21.2 & 12.6 & 22.8 & 21.2 & 12.6 & 22.8 & 21.2 & 12.6 & 22.8 & 21.2 & 12.6 \\
\cmidrule{2-17}
 & PCA & 23.2 & 21.7 & 13.1 & \textbf{22.8} & \textbf{21.2} & 12.6 & \textbf{22.8} & \textbf{21.2} & 12.6 & \textbf{22.5} & \textbf{20.9} & 12.4 & \textbf{22.8} & \underline{21.2} & 12.6 \\
 & OPP & \underline{25.1} & \underline{24.8} & \textbf{17.1} & \underline{21.5} & \textbf{21.2} & \textbf{16.0} & \underline{20.8} & \underline{20.3} & \textbf{15.5} & 18.2 & 18.2 & \textbf{15.0} & 19.5 & 19.4 & \textbf{16.6} \\
 & EFO & \underline{25.1} & \underline{24.8} & \textbf{17.1} & \underline{21.5} & \textbf{21.2} & \textbf{16.0} & \underline{20.8} & \underline{20.3} & \textbf{15.5} & 18.2 & 18.2 & \textbf{15.0} & 19.5 & 19.4 & \textbf{16.6} \\
 & OPF & \textbf{25.2} & \textbf{24.9} & \underline{15.5} & 15.0 & 14.6 & 8.5 & 15.0 & 14.2 & 8.3 & 16.4 & 16.3 & 14.0 & 18.7 & 18.6 & \underline{16.2} \\
 & OPR & \underline{25.1} & \underline{24.8} & \textbf{17.1} & \underline{21.5} & \textbf{21.2} & \textbf{16.0} & \underline{20.8} & \underline{20.3} & \textbf{15.5} & 18.2 & 18.2 & \textbf{15.0} & 19.5 & 19.4 & \textbf{16.6} \\
 & MOPP & 14.1 & 13.9 & 10.0 & 18.1 & \underline{17.9} & \underline{12.7} & 19.8 & 19.4 & \underline{13.9} & \underline{21.6} & \underline{20.8} & \underline{14.2} & \underline{22.6} & \textbf{21.7} & 15.1 \\
 & OPST & 14.1 & 13.9 & 10.0 & 18.1 & \underline{17.9} & \underline{12.7} & 19.8 & 19.4 & \underline{13.9} & \underline{21.6} & \underline{20.8} & \underline{14.2} & \underline{22.6} & \textbf{21.7} & 15.1 \\
 & SOPP & 11.6 & 11.4 & 6.8 & 11.4 & 11.2 & 7.0 & 11.2 & 11.1 & 6.7 & 14.1 & 13.9 & 9.3 & 12.6 & 12.4 & 8.0 \\
 & \CCG CT & \CCG \underline{25.1} & \CCG \underline{24.8} & \CCG \textbf{17.1} & \CCG \underline{21.5} & \CCG \textbf{21.2} & \CCG \textbf{16.0} & \CCG \underline{20.8} & \CCG \underline{20.3} & \CCG \textbf{15.5} & \CCG 18.2 & \CCG 18.2 & \CCG \textbf{15.0} & \CCG 19.5 & \CCG 19.4 & \CCG \textbf{16.6} \\
\midrule
\multirow{10}{*}{Beef} & Raw & 30.3 & 27.3 & 8.4 & 30.3 & 27.3 & 8.4 & 30.3 & 27.3 & 8.4 & 30.3 & 27.3 & 8.4 & 30.3 & 27.3 & 8.4 \\
\cmidrule{2-17}
 & PCA & \textbf{30.3} & 27.3 & \underline{8.4} & \underline{30.3} & 27.3 & 8.4 & \underline{30.3} & 27.3 & 8.4 & 30.3 & 27.3 & 8.4 & \textbf{30.3} & 27.3 & 8.4 \\
 & OPP & \underline{29.9} & \textbf{29.5} & \textbf{13.1} & 29.9 & \underline{29.0} & \underline{14.0} & 28.8 & \underline{27.8} & \underline{13.7} & 30.4 & \underline{29.6} & \underline{13.5} & 29.9 & \underline{29.0} & 12.6 \\
 & EFO & \underline{29.9} & \textbf{29.5} & \textbf{13.1} & 29.9 & \underline{29.0} & \underline{14.0} & 28.8 & \underline{27.8} & \underline{13.7} & 30.4 & \underline{29.6} & \underline{13.5} & 29.9 & \underline{29.0} & 12.6 \\
 & OPF & 29.7 & \underline{28.9} & \textbf{13.1} & \textbf{33.2} & \textbf{31.9} & \textbf{16.5} & \textbf{33.0} & \textbf{31.8} & \textbf{16.2} & \textbf{32.8} & \textbf{31.9} & \textbf{15.2} & \underline{30.1} & \textbf{29.2} & \textbf{13.7} \\
 & OPR & \underline{29.9} & \textbf{29.5} & \textbf{13.1} & 29.9 & \underline{29.0} & \underline{14.0} & 28.8 & \underline{27.8} & \underline{13.7} & 30.4 & \underline{29.6} & \underline{13.5} & 29.9 & \underline{29.0} & 12.6 \\
 & MOPP & 21.9 & 20.3 & 7.0 & 23.0 & 21.5 & 8.7 & 21.9 & 20.5 & 7.9 & 18.7 & 17.4 & 5.7 & 15.2 & 14.9 & 3.6 \\
 & OPST & 21.9 & 20.3 & 7.0 & 23.0 & 21.5 & 8.7 & 21.9 & 20.5 & 7.9 & 18.7 & 17.4 & 5.7 & 15.2 & 14.9 & 3.6 \\
 & SOPP & 13.3 & 13.0 & 3.2 & 12.8 & 12.7 & 1.3 & 14.0 & 13.9 & 1.7 & 11.4 & 11.3 & 0.6 & 9.7 & 9.6 & -0.7 \\
 & \CCG CT & \CCG \underline{29.9} & \CCG \textbf{29.5} & \CCG \textbf{13.1} & \CCG 29.9 & \CCG \underline{29.0} & \CCG \underline{14.0} & \CCG 28.8 & \CCG \underline{27.8} & \CCG \underline{13.7} & \CCG \underline{30.6} & \CCG \underline{29.6} & \CCG 13.4 & \CCG 29.5 & \CCG 28.5 & \CCG \underline{12.8} \\
\midrule
\multirow{10}{*}{\begin{tabular}[t]{@{}l@{}}Electric\\Devices\end{tabular}} & Raw & 19.9 & 19.3 & 16.8 & 19.9 & 19.3 & 16.8 & 19.9 & 19.3 & 16.8 & 19.9 & 19.3 & 16.8 & 19.9 & 19.3 & 16.8 \\
\cmidrule{2-17}
 & PCA & 20.9 & 20.2 & \textbf{20.1} & 20.3 & 19.5 & \underline{18.5} & 20.1 & 19.4 & \underline{17.9} & 20.3 & 19.3 & \underline{19.8} & 20.2 & 19.4 & \underline{18.1} \\
 & OPP & \underline{33.3} & \textbf{30.7} & 12.2 & \textbf{31.7} & \underline{28.6} & 11.1 & 26.0 & 22.2 & 7.2 & 25.0 & 21.1 & 6.8 & 23.7 & 19.5 & 6.2 \\
 & EFO & \underline{33.3} & \textbf{30.7} & 12.2 & \textbf{31.7} & \underline{28.6} & 11.1 & 26.0 & 22.2 & 7.2 & 25.0 & 21.1 & 6.8 & 23.7 & 19.5 & 6.2 \\
 & OPF & 31.9 & 29.8 & 11.7 & \underline{30.1} & 27.3 & 10.1 & 25.2 & 21.5 & 6.7 & 24.0 & 20.2 & 6.4 & 22.2 & 18.5 & 5.5 \\
 & OPR & \underline{33.3} & \textbf{30.7} & 12.2 & \textbf{31.7} & \underline{28.6} & 11.1 & 26.0 & 22.2 & 7.2 & 25.0 & 21.1 & 6.8 & 23.7 & 19.5 & 6.2 \\
 & MOPP & \underline{33.3} & \textbf{30.7} & 12.2 & \textbf{31.7} & \underline{28.6} & 11.1 & 26.0 & 22.2 & 7.2 & 25.0 & 21.1 & 6.8 & 23.7 & 19.5 & 6.2 \\
 & OPST & 28.7 & 28.9 & \underline{16.5} & 28.8 & \textbf{28.7} & \textbf{20.4} & \textbf{28.6} & \textbf{28.8} & \textbf{21.0} & \textbf{28.6} & \textbf{28.9} & \textbf{21.3} & \textbf{28.2} & \textbf{28.5} & \textbf{22.4} \\
 & SOPP & \underline{33.3} & \textbf{30.7} & 12.2 & 27.2 & 23.6 & 7.6 & 26.0 & 22.2 & 7.2 & 24.4 & 20.5 & 6.5 & 23.1 & 18.9 & 6.1 \\
 & \CCG CT & \CCG \textbf{33.4} & \CCG \underline{29.9} & \CCG 13.4 & \CCG 27.3 & \CCG 23.4 & \CCG 7.7 & \CCG \underline{26.6} & \CCG \underline{22.6} & \CCG 7.4 & \CCG \underline{25.6} & \CCG \underline{21.7} & \CCG 7.2 & \CCG \underline{24.3} & \CCG \underline{20.1} & \CCG 6.3 \\
\midrule
\multirow{10}{*}{MoteStrain} & Raw & 40.9 & 40.9 & 51.0 & 40.9 & 40.9 & 51.0 & 40.9 & 40.9 & 51.0 & 40.9 & 40.9 & 51.0 & 40.9 & 40.9 & 51.0 \\
\cmidrule{2-17}
 & PCA & 40.8 & 40.8 & 50.9 & 40.7 & 40.8 & 50.9 & 40.7 & 40.8 & \textbf{50.9} & 40.7 & 40.7 & \textbf{50.8} & \underline{40.8} & \textbf{40.8} & \textbf{50.9} \\
 & OPP & \textbf{53.6} & \textbf{53.7} & \textbf{64.3} & \underline{50.8} & \underline{50.8} & \underline{58.1} & \underline{47.1} & \underline{46.5} & 48.0 & \underline{43.7} & \underline{42.7} & 41.5 & 39.8 & 38.5 & 35.6 \\
 & EFO & \textbf{53.6} & \textbf{53.7} & \textbf{64.3} & \underline{50.8} & \underline{50.8} & \underline{58.1} & \underline{47.1} & \underline{46.5} & 48.0 & \underline{43.7} & \underline{42.7} & 41.5 & 39.8 & 38.5 & 35.6 \\
 & OPF & 45.4 & 45.4 & 56.0 & 45.8 & 45.8 & 53.4 & 42.3 & 41.8 & 44.0 & 39.4 & 38.1 & 35.5 & 36.0 & 34.3 & 29.3 \\
 & OPR & \textbf{53.6} & \textbf{53.7} & \textbf{64.3} & \underline{50.8} & \underline{50.8} & \underline{58.1} & \underline{47.1} & \underline{46.5} & 48.0 & \underline{43.7} & \underline{42.7} & 41.5 & 39.8 & 38.5 & 35.6 \\
 & MOPP & \textbf{53.6} & \textbf{53.7} & \textbf{64.3} & \underline{50.8} & \underline{50.8} & \underline{58.1} & \underline{47.1} & \underline{46.5} & 48.0 & \underline{43.7} & \underline{42.7} & 41.5 & 39.8 & 38.5 & 35.6 \\
 & OPST & \underline{53.4} & \underline{53.5} & \underline{62.1} & 49.1 & 48.8 & 52.7 & 45.5 & 44.8 & 45.7 & 42.9 & 41.8 & 40.3 & 34.3 & 32.0 & 24.1 \\
 & SOPP & 12.7 & 12.7 & 16.2 & 31.1 & 29.8 & 27.6 & 35.5 & 33.9 & 29.7 & 34.2 & 32.3 & 26.6 & 32.9 & 30.9 & 24.7 \\
 & \CCG CT & \CCG \textbf{53.6} & \CCG \textbf{53.7} & \CCG \textbf{64.3} & \CCG \textbf{51.7} & \CCG \textbf{51.7} & \CCG \textbf{59.3} & \CCG \textbf{47.9} & \CCG \textbf{47.3} & \CCG \underline{49.1} & \CCG \textbf{44.2} & \CCG \textbf{43.3} & \CCG \underline{42.3} & \CCG \textbf{40.9} & \CCG \underline{39.7} & \CCG \underline{37.3} \\
\midrule
\multirow{10}{*}{PigCVP} & Raw & 55.9 & 52.6 & 5.6 & 55.9 & 52.6 & 5.6 & 55.9 & 52.6 & 5.6 & 55.9 & 52.6 & 5.6 & 55.9 & 52.6 & 5.6 \\
\cmidrule{2-17}
 & PCA & 58.2 & 56.6 & 5.5 & 58.2 & 55.8 & 6.4 & 57.6 & 55.0 & 6.0 & 57.1 & 54.0 & 5.9 & 56.8 & 53.6 & 6.0 \\
 & OPP & \underline{59.1} & \textbf{57.7} & \underline{6.2} & 62.4 & \underline{60.9} & 10.3 & \underline{63.8} & \underline{62.1} & 12.9 & \textbf{66.5} & \textbf{64.9} & \textbf{16.9} & \textbf{67.2} & \textbf{65.8} & \textbf{18.1} \\
 & EFO & \underline{59.1} & \textbf{57.7} & \underline{6.2} & 62.4 & \underline{60.9} & 10.3 & \underline{63.8} & \underline{62.1} & 12.9 & \textbf{66.5} & \textbf{64.9} & \textbf{16.9} & \textbf{67.2} & \textbf{65.8} & \textbf{18.1} \\
 & OPF & 58.9 & \textbf{57.7} & 5.8 & 61.2 & 59.5 & 9.0 & 63.6 & 62.0 & 12.4 & 65.8 & 64.3 & 15.5 & 66.3 & 64.6 & 16.3 \\
 & OPR & \underline{59.1} & \textbf{57.7} & \underline{6.2} & 62.4 & \underline{60.9} & 10.3 & \underline{63.8} & \underline{62.1} & 12.9 & \textbf{66.5} & \textbf{64.9} & \textbf{16.9} & \textbf{67.2} & \textbf{65.8} & \textbf{18.1} \\
 & MOPP & \textbf{59.2} & 56.0 & \textbf{7.8} & \underline{62.6} & 60.8 & \underline{11.2} & \underline{63.8} & 61.8 & \underline{13.4} & 64.1 & 62.5 & 13.0 & 65.6 & 63.7 & 15.7 \\
 & OPST & \textbf{59.2} & 56.0 & \textbf{7.8} & \underline{62.6} & 60.8 & \underline{11.2} & \underline{63.8} & 61.8 & \underline{13.4} & 64.1 & 62.5 & 13.0 & 65.6 & 63.7 & 15.7 \\
 & SOPP & 58.5 & 57.1 & 5.6 & 58.9 & 57.6 & 5.7 & 58.3 & 56.8 & 5.3 & 58.6 & 57.1 & 5.9 & 58.9 & 57.3 & 6.4 \\
 & \CCG CT & \CCG 58.6 & \CCG \underline{57.2} & \CCG 5.7 & \CCG \textbf{63.1} & \CCG \textbf{61.5} & \CCG \textbf{11.5} & \CCG \textbf{64.7} & \CCG \textbf{63.1} & \CCG \textbf{14.1} & \CCG \underline{66.0} & \CCG \underline{64.4} & \CCG \underline{15.9} & \CCG \underline{66.7} & \CCG \underline{65.2} & \CCG \underline{17.2} \\
\bottomrule[1.2pt]
\end{tabular}
\end{table}

\begin{table}[th]\centering
\caption{
Feature-construction time and peak memory across feature budgets
on the five main datasets.
Time is averaged over five runs and memory is measured as peak process RSS.
Bold and underline denote the lowest and second-lowest values excluding Raw.
}
\label{CTPM:app:tab:budget_resources}
\small
\setlength{\tabcolsep}{4.5pt}
\begin{tabular}{llrrrrrrrrrr}
\toprule[1.2pt]
\multirow{2}{*}[-2pt]{Dataset} & \multirow{2}{*}[-2pt]{Method} & \multicolumn{2}{c}{$K=10$} & \multicolumn{2}{c}{$K=20$} & \multicolumn{2}{c}{$K=30$} & \multicolumn{2}{c}{$K=40$} & \multicolumn{2}{c}{$K=50$} \\
\cmidrule(lr){3-4}\cmidrule(lr){5-6}\cmidrule(lr){7-8}\cmidrule(lr){9-10}\cmidrule(lr){11-12}
 &  & Time. & Mem. & Time. & Mem. & Time. & Mem. & Time. & Mem. & Time. & Mem. \\
\midrule[1.2pt]
\multirow{10}{*}{Car} & Raw & 0.000 & 158.4 & 0.000 & 158.4 & 0.000 & 158.4 & 0.000 & 158.4 & 0.000 & 158.4 \\
\cmidrule{2-12}
 & PCA & \textbf{0.199} & \underline{163.7} & \textbf{0.201} & \underline{159.4} & \textbf{0.205} & \underline{163.5} & \textbf{0.208} & \textbf{161.2} & \textbf{0.203} & \textbf{159.7} \\
 & OPP & 20.57 & 650.0 & 20.82 & 653.0 & 20.78 & 650.7 & 20.36 & 648.0 & 20.60 & 651.0 \\
 & EFO & 17.07 & 636.7 & 17.01 & 639.2 & 17.09 & 636.0 & 17.27 & 636.5 & 17.22 & 635.1 \\
 & OPF & 15.57 & 633.4 & 15.69 & 633.5 & 15.05 & 632.7 & 15.72 & 632.1 & 15.67 & 634.6 \\
 & OPR & 17.54 & 639.9 & 17.53 & 639.3 & 17.76 & 642.3 & 17.66 & 640.8 & 17.92 & 639.9 \\
 & MOPP & 18.68 & 644.6 & 18.73 & 645.1 & 18.81 & 640.1 & 18.61 & 641.7 & 18.72 & 639.2 \\
 & OPST & 14.09 & 628.8 & 14.23 & 625.7 & 13.90 & 619.2 & 14.12 & 625.8 & 13.82 & 628.5 \\
 & SOPP & 15.81 & 626.9 & 15.64 & 631.9 & 15.71 & 634.0 & 15.94 & 630.7 & 16.03 & 632.2 \\
 & \CCG CT & \CCG \underline{0.342} & \CCG \textbf{159.0} & \CCG \underline{0.442} & \CCG \textbf{158.1} & \CCG \underline{0.543} & \CCG \textbf{158.8} & \CCG \underline{0.633} & \CCG \underline{184.3} & \CCG \underline{0.753} & \CCG \underline{228.3} \\
\midrule
\multirow{10}{*}{Beef} & Raw & 0.000 & 156.4 & 0.000 & 156.4 & 0.000 & 156.4 & 0.000 & 156.4 & 0.000 & 156.4 \\
\cmidrule{2-12}
 & PCA & \textbf{0.003} & \underline{162.4} & \textbf{0.004} & \underline{163.3} & \textbf{0.004} & \underline{161.2} & \textbf{0.005} & \underline{158.0} & \textbf{0.005} & \underline{158.3} \\
 & OPP & 2.94 & 640.1 & 3.04 & 641.2 & 3.05 & 639.5 & 3.04 & 645.5 & 3.04 & 637.9 \\
 & EFO & 2.58 & 620.7 & 2.55 & 613.9 & 2.57 & 620.8 & 2.58 & 621.6 & 2.52 & 621.8 \\
 & OPF & 2.01 & 622.0 & 2.05 & 619.7 & 2.05 & 620.5 & 2.04 & 624.1 & 2.03 & 618.0 \\
 & OPR & 2.68 & 621.5 & 2.74 & 622.1 & 2.68 & 622.7 & 2.70 & 622.0 & 2.71 & 621.2 \\
 & MOPP & 2.73 & 651.2 & 2.71 & 636.4 & 2.70 & 635.3 & 2.82 & 635.7 & 2.68 & 641.3 \\
 & OPST & 1.69 & 619.2 & 1.68 & 615.6 & 1.65 & 618.7 & 1.65 & 618.0 & 1.67 & 635.1 \\
 & SOPP & 2.04 & 617.5 & 2.04 & 620.5 & 2.10 & 619.5 & 2.06 & 625.0 & 2.03 & 618.3 \\
 & \CCG CT & \CCG \underline{0.257} & \CCG \textbf{156.7} & \CCG \underline{0.309} & \CCG \textbf{156.7} & \CCG \underline{0.378} & \CCG \textbf{157.0} & \CCG \underline{0.439} & \CCG \textbf{157.0} & \CCG \underline{0.491} & \CCG \textbf{156.6} \\
\midrule
\multirow{10}{*}{\begin{tabular}[t]{@{}l@{}}Electric\\Devices\end{tabular}} & Raw & 0.000 & 243.8 & 0.000 & 243.8 & 0.000 & 243.8 & 0.000 & 243.8 & 0.000 & 243.8 \\
\cmidrule{2-12}
 & PCA & \textbf{0.063} & \textbf{243.9} & \textbf{0.078} & \textbf{243.8} & \textbf{0.089} & \textbf{244.1} & \textbf{0.121} & \textbf{243.9} & \textbf{0.133} & \textbf{244.0} \\
 & OPP & 5.05 & 644.9 & 4.91 & 666.4 & 4.70 & 667.5 & 4.91 & 666.5 & 4.83 & 667.1 \\
 & EFO & 4.38 & 649.2 & 4.33 & 647.9 & 4.38 & 670.7 & 4.31 & 646.3 & 4.39 & 647.1 \\
 & OPF & 4.76 & 669.3 & 4.25 & 665.4 & 4.32 & 665.1 & 4.46 & 658.5 & 4.31 & 653.4 \\
 & OPR & 4.65 & 650.9 & 4.55 & 658.4 & 4.62 & 650.0 & 4.55 & 655.4 & 4.64 & 648.3 \\
 & MOPP & 4.76 & 666.3 & 4.78 & 664.9 & 5.02 & 664.7 & 4.93 & 664.9 & 5.02 & 666.9 \\
 & OPST & 45.22 & 676.7 & 45.28 & 672.3 & 45.58 & 678.8 & 45.78 & 677.6 & 45.36 & 677.7 \\
 & SOPP & 3.75 & \underline{642.8} & 3.69 & \underline{639.9} & 3.67 & \underline{642.7} & 3.71 & \underline{644.0} & 3.78 & \underline{643.7} \\
 & \CCG CT & \CCG \underline{2.47} & \CCG 789.1 & \CCG \underline{2.58} & \CCG 785.7 & \CCG \underline{2.62} & \CCG 775.7 & \CCG \underline{2.56} & \CCG 783.9 & \CCG \underline{2.43} & \CCG 781.1 \\
\midrule
\multirow{10}{*}{MoteStrain} & Raw & 0.000 & 160.8 & 0.000 & 160.8 & 0.000 & 160.8 & 0.000 & 160.8 & 0.000 & 160.8 \\
\cmidrule{2-12}
 & PCA & \textbf{0.005} & \textbf{160.8} & \textbf{0.007} & \textbf{161.8} & \textbf{0.008} & \textbf{165.6} & \textbf{0.009} & \textbf{161.7} & \textbf{0.011} & \textbf{162.3} \\
 & OPP & 4.34 & 644.3 & 4.29 & 643.3 & 4.30 & 640.3 & 4.25 & 642.7 & 4.34 & 644.8 \\
 & EFO & 3.75 & 622.5 & 3.79 & 623.3 & 3.70 & 626.6 & 3.65 & 626.5 & 3.70 & 625.3 \\
 & OPF & 3.62 & 642.4 & 3.54 & 635.2 & 3.54 & 635.1 & 3.56 & 639.7 & 3.57 & 637.4 \\
 & OPR & 4.11 & 649.1 & 4.04 & 624.7 & 4.05 & 650.7 & 4.13 & 650.5 & 4.18 & 650.5 \\
 & MOPP & 4.24 & 635.4 & 4.24 & 632.6 & 4.32 & 637.2 & 4.23 & 637.3 & 4.24 & 642.2 \\
 & OPST & 3.36 & 624.7 & 3.28 & 621.1 & 3.36 & 627.9 & 3.32 & 624.8 & 3.34 & 632.0 \\
 & SOPP & 3.34 & 627.6 & 3.41 & 620.3 & 3.31 & 632.6 & 3.35 & 620.5 & 3.29 & 621.0 \\
 & \CCG CT & \CCG \underline{0.573} & \CCG \underline{173.2} & \CCG \underline{0.754} & \CCG \underline{227.7} & \CCG \underline{0.994} & \CCG \underline{308.1} & \CCG \underline{1.21} & \CCG \underline{353.9} & \CCG \underline{2.74} & \CCG \underline{390.8} \\
\midrule
\multirow{10}{*}{PigCVP} & Raw & 0.000 & 188.1 & 0.000 & 188.1 & 0.000 & 188.1 & 0.000 & 188.1 & 0.000 & 188.1 \\
\cmidrule{2-12}
 & PCA & \textbf{0.023} & \textbf{188.5} & \textbf{0.028} & \textbf{188.2} & \textbf{0.032} & \textbf{188.4} & \textbf{0.029} & \textbf{188.1} & \textbf{0.032} & \textbf{187.8} \\
 & OPP & 42.84 & 658.8 & 42.67 & 663.3 & 43.19 & 664.5 & 42.85 & 659.6 & 42.61 & 673.3 \\
 & EFO & 33.40 & 652.7 & 33.55 & 648.6 & 33.12 & 649.0 & 33.49 & 648.8 & 33.08 & 649.6 \\
 & OPF & 17.89 & 654.8 & 17.82 & 652.5 & 17.78 & 649.9 & 18.01 & 655.0 & 17.72 & 651.4 \\
 & OPR & 34.24 & 653.7 & 34.05 & 650.7 & 34.41 & 662.9 & 34.00 & 656.1 & 34.24 & 655.5 \\
 & MOPP & 37.92 & 654.2 & 37.46 & 666.1 & 37.31 & 662.1 & 36.87 & 673.3 & 36.97 & 665.2 \\
 & OPST & 19.07 & 637.7 & 18.82 & 639.6 & 18.65 & \underline{645.4} & 18.69 & \underline{643.4} & 18.89 & \underline{641.6} \\
 & SOPP & 18.69 & 660.9 & 18.40 & 657.2 & 18.35 & 660.5 & 18.04 & 657.3 & 18.48 & 659.4 \\
 & \CCG CT & \CCG \underline{1.77} & \CCG \underline{557.0} & \CCG \underline{1.58} & \CCG \underline{570.0} & \CCG \underline{1.94} & \CCG 661.6 & \CCG \underline{2.59} & \CCG 673.8 & \CCG \underline{3.09} & \CCG 683.6 \\
\bottomrule[1.2pt]
\end{tabular}
\end{table}
\FloatBarrier

\subsection{Full Results with Increasing Input Length}

We vary the retained input prefix from $20\%$ to $100\%$
while fixing the feature budget at $K=30$.
Tables~\ref{CTPM:app:tab:prefix_quality} and~\ref{CTPM:app:tab:prefix_resources}
report the complete dataset-level results,
while Figure~\ref{CTPM:app:fig:prefix_resources} summarizes
the corresponding resource trends across the five datasets.
PCA with $K=30$ is infeasible at the $20\%$ prefix
for ElectricDevices and MoteStrain because the retained
sequence length is smaller than the requested dimension.

\begin{figure}[htbp]\centering
\includegraphics[width=\linewidth]{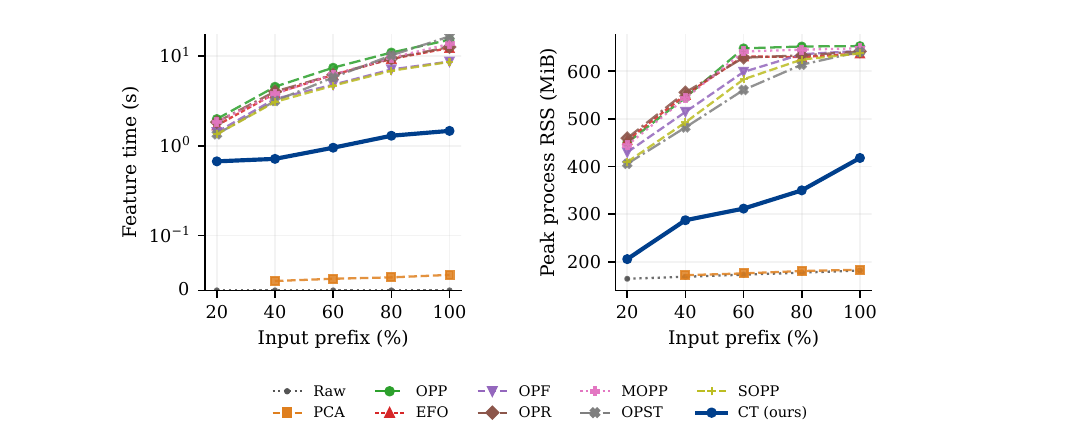}
\caption{
Mean feature-construction time and peak process RSS
as the retained input prefix increases at $K=30$.
Results are averaged over the five main datasets.
}
\label{CTPM:app:fig:prefix_resources}\end{figure}

Figure~\ref{CTPM:app:fig:prefix_resources} shows that CT maintains
a substantially lower feature-construction time than the other
pattern-mining methods as the input grows.
Its runtime increases gradually with the retained prefix,
whereas the gap to the OP-based miners generally widens.
Mean memory usage also remains lower for CT,
although its footprint increases more visibly with input length.
The full dataset-level results show that this memory advantage
is less uniform than the runtime advantage.

The quality results are strongly dataset dependent
and do not vary monotonically with input length.
PigCVP generally benefits from longer prefixes,
with CT improving steadily from the shorter inputs to the full sequence.
In contrast, Car and Beef fluctuate across prefixes,
and MoteStrain shows particularly large changes at intermediate lengths.
ElectricDevices again exhibits a distinct behavior,
with OPST achieving the strongest clustering quality over most prefixes.
These results indicate that longer inputs do not necessarily yield
better clustering representations, even when more temporal information
is available.

The dataset-level resource results reinforce the runtime trend
while revealing more variation in memory usage.
CT remains among the fastest pattern miners across essentially
all prefix settings and shows especially modest runtime growth
on longer sequences such as PigCVP.
Its memory footprint is very small on Car and Beef
and remains favorable on much of MoteStrain,
but grows substantially on ElectricDevices and PigCVP
as the retained input becomes longer.
Thus, the input-length sweep confirms a consistent runtime advantage
while also exposing the memory cost of retaining shared
Cartesian structures on more demanding inputs.

\begin{table}[th]\centering
\caption{
Clustering quality across input prefixes at $K=30$
on the five main datasets.
Values are mean NMI, homogeneity, and ARI over five clustering seeds.
Bold and underline denote the best and second-best values excluding Raw.
}
\label{CTPM:app:tab:prefix_quality}
\small
\setlength{\tabcolsep}{2pt}
\begin{tabular}{llrrrrrrrrrrrrrrr}
\toprule[1.2pt]
\multirow{2}{*}[-2pt]{Dataset} & \multirow{2}{*}[-2pt]{Method} & \multicolumn{3}{c}{20\%} & \multicolumn{3}{c}{40\%} & \multicolumn{3}{c}{60\%} & \multicolumn{3}{c}{80\%} & \multicolumn{3}{c}{100\%} \\
\cmidrule(lr){3-5}\cmidrule(lr){6-8}\cmidrule(lr){9-11}\cmidrule(lr){12-14}\cmidrule(lr){15-17}
 &  & NMI & $h$ & ARI & NMI & $h$ & ARI & NMI & $h$ & ARI & NMI & $h$ & ARI & NMI & $h$ & ARI \\
\midrule[1.2pt]
\multirow{10}{*}{Car} & Raw & 17.7 & 16.3 & 8.4 & 20.8 & 19.6 & 11.1 & 26.5 & 24.5 & 15.1 & 21.6 & 19.9 & 11.1 & 22.8 & 21.2 & 12.6 \\
\cmidrule{2-17}
 & PCA & 17.7 & 16.3 & 8.4 & \textbf{20.8} & \textbf{19.6} & \underline{11.1} & \textbf{26.2} & \underline{24.2} & 14.8 & \textbf{21.6} & \textbf{19.9} & 11.1 & \textbf{22.8} & \textbf{21.2} & 12.6 \\
 & OPP & 17.8 & 16.9 & 9.1 & 14.2 & 14.1 & \textbf{11.7} & 13.3 & 13.1 & 6.8 & 14.4 & 14.3 & 9.8 & \underline{20.8} & \underline{20.3} & \textbf{15.5} \\
 & EFO & 17.8 & 16.9 & 9.1 & 14.2 & 14.1 & \textbf{11.7} & 13.3 & 13.1 & 6.8 & 14.4 & 14.3 & 9.8 & \underline{20.8} & \underline{20.3} & \textbf{15.5} \\
 & OPF & 16.4 & 15.8 & 8.7 & 8.8 & 8.7 & 3.5 & \underline{26.0} & \textbf{25.6} & \textbf{21.5} & 11.9 & 11.8 & 7.6 & 15.0 & 14.2 & 8.3 \\
 & OPR & 17.8 & 16.9 & 9.1 & 14.2 & 14.1 & \textbf{11.7} & 13.3 & 13.1 & 6.8 & 14.4 & 14.3 & 9.8 & \underline{20.8} & \underline{20.3} & \textbf{15.5} \\
 & MOPP & \textbf{19.2} & \textbf{18.2} & \underline{9.9} & 15.3 & 14.8 & 7.4 & 25.5 & 24.1 & \underline{17.3} & \underline{17.9} & \underline{16.9} & \textbf{12.5} & 19.8 & 19.4 & \underline{13.9} \\
 & OPST & \textbf{19.2} & \textbf{18.2} & \underline{9.9} & 15.3 & 14.8 & 7.4 & 25.5 & 24.1 & \underline{17.3} & \underline{17.9} & \underline{16.9} & \textbf{12.5} & 19.8 & 19.4 & \underline{13.9} \\
 & SOPP & \underline{18.8} & \underline{18.0} & \textbf{10.7} & \underline{17.2} & \underline{16.6} & 8.6 & 14.5 & 14.3 & 12.3 & 14.2 & 14.1 & \underline{11.9} & 11.2 & 11.1 & 6.7 \\
 & \CCG CT & \CCG 17.8 & \CCG 16.9 & \CCG 9.1 & \CCG 14.2 & \CCG 14.1 & \CCG \textbf{11.7} & \CCG 13.3 & \CCG 13.1 & \CCG 6.8 & \CCG 14.4 & \CCG 14.3 & \CCG 9.8 & \CCG \underline{20.8} & \CCG \underline{20.3} & \CCG \textbf{15.5} \\
\midrule
\multirow{10}{*}{Beef} & Raw & 31.9 & 31.3 & 14.1 & 41.7 & 40.8 & 22.7 & 33.1 & 30.6 & 12.2 & 31.9 & 29.0 & 10.0 & 30.3 & 27.3 & 8.4 \\
\cmidrule{2-17}
 & PCA & \textbf{31.9} & \textbf{31.3} & \textbf{14.1} & \textbf{41.7} & \textbf{40.8} & \textbf{22.7} & \textbf{33.1} & \textbf{30.6} & \textbf{12.2} & \textbf{31.9} & \textbf{29.0} & \textbf{10.0} & \underline{30.3} & 27.3 & 8.4 \\
 & OPP & 17.4 & 17.3 & 4.4 & \underline{24.7} & \underline{24.0} & \underline{9.4} & 20.3 & 18.8 & 6.2 & 16.4 & 15.9 & 2.4 & 28.8 & \underline{27.8} & \underline{13.7} \\
 & EFO & 17.4 & 17.3 & 4.4 & \underline{24.7} & \underline{24.0} & \underline{9.4} & 20.3 & 18.8 & 6.2 & 16.4 & 15.9 & 2.4 & 28.8 & \underline{27.8} & \underline{13.7} \\
 & OPF & 18.1 & 17.8 & 4.6 & 23.9 & 23.3 & 8.8 & \underline{20.8} & 19.5 & \underline{7.7} & \underline{24.2} & \underline{23.6} & \underline{8.4} & \textbf{33.0} & \textbf{31.8} & \textbf{16.2} \\
 & OPR & 17.4 & 17.3 & 4.4 & \underline{24.7} & \underline{24.0} & \underline{9.4} & 20.3 & 18.8 & 6.2 & 16.4 & 15.9 & 2.4 & 28.8 & \underline{27.8} & \underline{13.7} \\
 & MOPP & \underline{18.4} & \underline{18.0} & \underline{6.5} & 18.9 & 18.3 & 5.6 & \underline{20.8} & \underline{20.1} & 7.3 & 14.5 & 14.1 & 3.4 & 21.9 & 20.5 & 7.9 \\
 & OPST & \underline{18.4} & \underline{18.0} & \underline{6.5} & 18.9 & 18.3 & 5.6 & \underline{20.8} & \underline{20.1} & 7.3 & 14.5 & 14.1 & 3.4 & 21.9 & 20.5 & 7.9 \\
 & SOPP & 12.6 & 12.2 & 2.5 & 10.8 & 10.3 & -0.4 & 12.0 & 11.8 & 1.4 & 13.0 & 12.9 & 1.8 & 14.0 & 13.9 & 1.7 \\
 & \CCG CT & \CCG 17.4 & \CCG 17.3 & \CCG 4.4 & \CCG \underline{24.7} & \CCG \underline{24.0} & \CCG \underline{9.4} & \CCG 20.3 & \CCG 18.8 & \CCG 6.2 & \CCG 16.4 & \CCG 15.9 & \CCG 2.4 & \CCG 28.8 & \CCG \underline{27.8} & \CCG \underline{13.7} \\
\midrule
\multirow{10}{*}{\begin{tabular}[t]{@{}l@{}}Electric\\Devices\end{tabular}} & Raw & 18.0 & 15.0 & 0.6 & 15.8 & 13.8 & 4.7 & 15.6 & 14.7 & 10.4 & 12.8 & 11.9 & 9.6 & 19.9 & 19.3 & 16.8 \\
\cmidrule{2-17}
 & PCA & -- & -- & -- & 16.1 & 14.1 & 5.0 & 19.2 & 18.3 & \underline{19.0} & 12.9 & 12.2 & \underline{9.2} & 20.1 & 19.4 & \underline{17.9} \\
 & OPP & 28.1 & 23.6 & 9.3 & 26.5 & 21.8 & 8.5 & 24.5 & 20.4 & 7.6 & 24.2 & 20.6 & 7.1 & 26.0 & 22.2 & 7.2 \\
 & EFO & 28.1 & 23.6 & 9.3 & 26.5 & 21.8 & 8.5 & 24.5 & 20.4 & 7.6 & 24.2 & 20.6 & 7.1 & 26.0 & 22.2 & 7.2 \\
 & OPF & 28.5 & 23.8 & 9.6 & 24.3 & 20.4 & 8.6 & 22.8 & 19.1 & 7.6 & 23.0 & 19.6 & 6.9 & 25.2 & 21.5 & 6.7 \\
 & OPR & 28.1 & 23.6 & 9.3 & 26.5 & 21.8 & 8.5 & 24.5 & 20.4 & 7.6 & 24.2 & 20.6 & 7.1 & 26.0 & 22.2 & 7.2 \\
 & MOPP & 28.1 & 23.6 & 9.3 & 26.5 & 21.8 & 8.5 & 24.5 & 20.4 & 7.6 & 24.2 & 20.6 & 7.1 & 26.0 & 22.2 & 7.2 \\
 & OPST & \textbf{33.8} & \textbf{31.7} & \textbf{10.8} & \textbf{29.9} & \textbf{29.0} & \textbf{20.6} & \textbf{30.0} & \textbf{30.0} & \textbf{22.2} & \textbf{29.3} & \textbf{29.5} & \textbf{21.0} & \textbf{28.6} & \textbf{28.8} & \textbf{21.0} \\
 & SOPP & 28.5 & 23.9 & 9.6 & \underline{27.5} & 23.1 & 9.0 & 24.5 & 20.4 & 7.6 & 24.2 & 20.6 & 7.1 & 26.0 & 22.2 & 7.2 \\
 & \CCG CT & \CCG \underline{28.9} & \CCG \underline{24.4} & \CCG \underline{9.7} & \CCG \underline{27.5} & \CCG \underline{23.2} & \CCG \underline{9.4} & \CCG \underline{25.3} & \CCG \underline{21.3} & \CCG 7.9 & \CCG \underline{25.1} & \CCG \underline{21.3} & \CCG 7.5 & \CCG \underline{26.6} & \CCG \underline{22.6} & \CCG 7.4 \\
\midrule
\multirow{10}{*}{MoteStrain} & Raw & 30.5 & 29.8 & 37.5 & 29.1 & 29.0 & 37.7 & 43.4 & 43.5 & 52.8 & 43.7 & 43.8 & 53.7 & 40.9 & 40.9 & 51.0 \\
\cmidrule{2-17}
 & PCA & -- & -- & -- & \textbf{29.1} & \textbf{29.0} & \textbf{37.7} & \textbf{43.5} & \textbf{43.6} & \textbf{52.8} & \textbf{43.9} & \textbf{44.0} & \textbf{53.9} & 40.7 & 40.8 & \textbf{50.9} \\
 & OPP & 34.4 & 33.0 & 29.8 & 0.0 & 0.0 & 0.2 & \underline{33.3} & \underline{32.0} & \underline{29.8} & \underline{42.9} & \underline{42.4} & \underline{45.4} & \underline{47.1} & \underline{46.5} & 48.0 \\
 & EFO & 34.4 & 33.0 & 29.8 & 0.0 & 0.0 & 0.2 & \underline{33.3} & \underline{32.0} & \underline{29.8} & \underline{42.9} & \underline{42.4} & \underline{45.4} & \underline{47.1} & \underline{46.5} & 48.0 \\
 & OPF & 30.5 & 28.3 & 21.7 & 0.0 & 0.0 & -0.0 & 25.3 & 22.8 & 15.5 & 38.6 & 37.9 & 39.1 & 42.3 & 41.8 & 44.0 \\
 & OPR & 34.4 & 33.0 & 29.8 & 0.0 & 0.0 & 0.2 & \underline{33.3} & \underline{32.0} & \underline{29.8} & \underline{42.9} & \underline{42.4} & \underline{45.4} & \underline{47.1} & \underline{46.5} & 48.0 \\
 & MOPP & 34.4 & 33.0 & 29.8 & 0.0 & 0.0 & 0.2 & \underline{33.3} & \underline{32.0} & \underline{29.8} & \underline{42.9} & \underline{42.4} & \underline{45.4} & \underline{47.1} & \underline{46.5} & 48.0 \\
 & OPST & 34.4 & 33.0 & 29.8 & 0.0 & 0.0 & -0.0 & 29.5 & 27.3 & 20.8 & 42.1 & 41.3 & 41.5 & 45.5 & 44.8 & 45.7 \\
 & SOPP & \underline{35.9} & \underline{34.3} & \underline{30.2} & \underline{0.2} & \underline{0.2} & \underline{0.5} & 0.0 & 0.0 & 0.1 & 39.9 & 39.5 & 43.6 & 35.5 & 33.9 & 29.7 \\
 & \CCG CT & \CCG \textbf{37.2} & \CCG \textbf{37.0} & \CCG \textbf{41.5} & \CCG 0.0 & \CCG 0.0 & \CCG -0.1 & \CCG 33.1 & \CCG 31.6 & \CCG 28.0 & \CCG 42.5 & \CCG 42.0 & \CCG 44.8 & \CCG \textbf{47.9} & \CCG \textbf{47.3} & \CCG \underline{49.1} \\
\midrule
\multirow{10}{*}{PigCVP} & Raw & 55.9 & 52.8 & 5.0 & 56.0 & 52.8 & 5.2 & 56.5 & 53.7 & 5.3 & 56.3 & 53.2 & 5.6 & 55.9 & 52.6 & 5.6 \\
\cmidrule{2-17}
 & PCA & \textbf{56.6} & 53.8 & \textbf{5.4} & 56.2 & 53.2 & \underline{5.2} & 57.5 & 54.7 & 6.1 & 57.0 & 54.1 & 5.7 & 57.6 & 55.0 & 6.0 \\
 & OPP & 55.3 & 53.9 & 2.2 & 56.8 & 55.4 & 4.0 & 60.6 & 58.9 & 8.7 & \underline{62.7} & \underline{61.1} & 11.3 & \underline{63.8} & \underline{62.1} & 12.9 \\
 & EFO & 55.3 & 53.9 & 2.2 & 56.8 & 55.4 & 4.0 & 60.6 & 58.9 & 8.7 & \underline{62.7} & \underline{61.1} & 11.3 & \underline{63.8} & \underline{62.1} & 12.9 \\
 & OPF & \underline{55.7} & \underline{54.2} & 2.7 & 57.8 & 56.4 & 4.8 & 60.7 & 58.9 & 8.7 & \underline{62.7} & \underline{61.1} & \underline{11.6} & 63.6 & 62.0 & 12.4 \\
 & OPR & 55.3 & 53.9 & 2.2 & 56.8 & 55.4 & 4.0 & 60.6 & 58.9 & 8.7 & \underline{62.7} & \underline{61.1} & 11.3 & \underline{63.8} & \underline{62.1} & 12.9 \\
 & MOPP & \textbf{56.6} & \textbf{55.3} & \underline{3.7} & \textbf{58.6} & \underline{56.9} & \textbf{6.0} & \underline{61.0} & \underline{59.1} & \underline{9.2} & 62.1 & 59.9 & 11.0 & \underline{63.8} & 61.8 & \underline{13.4} \\
 & OPST & \textbf{56.6} & \textbf{55.3} & \underline{3.7} & \textbf{58.6} & \underline{56.9} & \textbf{6.0} & \underline{61.0} & \underline{59.1} & \underline{9.2} & 62.1 & 59.9 & 11.0 & \underline{63.8} & 61.8 & \underline{13.4} \\
 & SOPP & 55.0 & 53.6 & 2.2 & 56.5 & 54.8 & 4.0 & 59.4 & 57.8 & 7.0 & 59.0 & 57.4 & 6.6 & 58.3 & 56.8 & 5.3 \\
 & \CCG CT & \CCG 55.4 & \CCG 54.1 & \CCG 2.0 & \CCG \underline{58.5} & \CCG \textbf{57.2} & \CCG \textbf{6.0} & \CCG \textbf{61.8} & \CCG \textbf{60.0} & \CCG \textbf{10.5} & \CCG \textbf{64.2} & \CCG \textbf{62.6} & \CCG \textbf{13.4} & \CCG \textbf{64.7} & \CCG \textbf{63.1} & \CCG \textbf{14.1} \\
\bottomrule[1.2pt]
\end{tabular}
\end{table}

\begin{table}[th]\centering
\caption{
Feature-construction time and peak memory across input prefixes
at $K=30$ on the five main datasets.
Bold and underline denote the lowest and second-lowest values excluding Raw.
}
\label{CTPM:app:tab:prefix_resources}
\small
\setlength{\tabcolsep}{4.5pt}
\begin{tabular}{llrrrrrrrrrr}\toprule[1.2pt]
\multirow{2}{*}[-2pt]{Dataset} & \multirow{2}{*}[-2pt]{Method} & \multicolumn{2}{c}{20\%} & \multicolumn{2}{c}{40\%} & \multicolumn{2}{c}{60\%} & \multicolumn{2}{c}{80\%} & \multicolumn{2}{c}{100\%} \\
\cmidrule(lr){3-4}\cmidrule(lr){5-6}\cmidrule(lr){7-8}\cmidrule(lr){9-10}\cmidrule(lr){11-12}
 &  & Time. & Mem. & Time. & Mem. & Time. & Mem. & Time. & Mem. & Time. & Mem. \\
\midrule[1.2pt]
\multirow{10}{*}{Car} & Raw & 0.000 & 157.1 & 0.000 & 157.0 & 0.000 & 157.3 & 0.000 & 158.3 & 0.000 & 158.9 \\
\cmidrule{2-12}
 & PCA & \textbf{0.004} & \underline{157.9} & \textbf{0.004} & \underline{162.9} & \textbf{0.005} & \underline{163.7} & \textbf{0.005} & \textbf{164.5} & \textbf{0.006} & \underline{160.4} \\
 & OPP & 1.99 & 619.9 & 5.25 & 637.7 & 8.70 & 647.6 & 15.58 & 648.3 & 20.86 & 652.7 \\
 & EFO & 1.73 & 618.1 & 4.67 & 622.1 & 7.40 & 627.6 & 13.60 & 626.9 & 17.36 & 632.7 \\
 & OPF & 1.63 & 620.6 & 4.31 & 621.2 & 7.12 & 630.5 & 12.53 & 640.2 & 15.71 & 634.8 \\
 & OPR & 1.87 & 618.2 & 4.75 & 622.0 & 7.60 & 630.6 & 13.79 & 627.6 & 17.90 & 640.9 \\
 & MOPP & 1.81 & 618.1 & 3.90 & 636.7 & 6.87 & 642.4 & 12.29 & 632.8 & 18.64 & 640.3 \\
 & OPST & 1.07 & 446.2 & 2.87 & 619.5 & 5.12 & 630.5 & 9.43 & 627.4 & 14.05 & 631.1 \\
 & SOPP & 1.40 & 581.8 & 4.34 & 620.3 & 7.15 & 629.7 & 12.91 & 627.1 & 15.88 & 633.1 \\
 & \CCG CT & \CCG \underline{0.372} & \CCG \textbf{156.5} & \CCG \underline{0.399} & \CCG \textbf{156.8} & \CCG \underline{0.473} & \CCG \textbf{158.0} & \CCG \underline{1.50} & \CCG \underline{165.9} & \CCG \underline{1.23} & \CCG \textbf{159.1} \\
\midrule
\multirow{10}{*}{Beef} & Raw & 0.000 & 156.2 & 0.000 & 157.0 & 0.000 & 156.3 & 0.000 & 156.8 & 0.000 & 157.0 \\
\cmidrule{2-12}
 & PCA & \textbf{0.003} & 161.5 & \textbf{0.003} & \underline{161.6} & \textbf{0.004} & \underline{157.1} & \textbf{0.004} & \underline{162.3} & \textbf{0.004} & \underline{158.1} \\
 & OPP & 0.680 & 206.1 & 1.27 & 421.9 & 1.95 & 644.7 & 2.19 & 647.6 & 3.05 & 642.4 \\
 & EFO & 0.689 & 219.7 & 1.15 & 456.3 & 1.57 & 617.8 & 1.95 & 617.4 & 2.58 & 620.5 \\
 & OPF & 0.618 & 194.0 & 0.922 & 317.4 & 1.24 & 446.9 & 1.49 & 586.4 & 2.09 & 624.5 \\
 & OPR & 0.730 & 231.1 & 1.21 & 463.3 & 1.68 & 619.4 & 2.04 & 623.6 & 2.73 & 623.5 \\
 & MOPP & 0.689 & 192.6 & 1.23 & 394.1 & 1.80 & 617.4 & 2.14 & 632.5 & 2.80 & 641.2 \\
 & OPST & \underline{0.448} & \textbf{156.1} & 0.754 & 260.5 & 0.984 & 364.9 & 1.27 & 506.0 & 1.65 & 617.9 \\
 & SOPP & 0.580 & 177.2 & 0.923 & 308.7 & 1.26 & 454.7 & 1.60 & 577.8 & 2.05 & 620.8 \\
 & \CCG CT & \CCG 1.07 & \CCG \underline{156.4} & \CCG \underline{0.290} & \CCG \textbf{156.2} & \CCG \underline{0.294} & \CCG \textbf{156.2} & \CCG \underline{0.356} & \CCG \textbf{156.4} & \CCG \underline{0.397} & \CCG \textbf{156.4} \\
\midrule
\multirow{10}{*}{\begin{tabular}[t]{@{}l@{}}Electric\\Devices\end{tabular}} & Raw & 0.000 & 184.9 & 0.000 & 199.5 & 0.000 & 215.1 & 0.000 & 227.7 & 0.000 & 243.8 \\
\cmidrule{2-12}
 & PCA & -- & -- & \textbf{0.058} & \textbf{199.5} & \textbf{0.071} & \textbf{215.2} & \textbf{0.077} & \textbf{227.6} & \textbf{0.092} & \textbf{244.0} \\
 & OPP & 1.89 & 621.2 & 2.32 & 631.2 & 3.11 & 657.5 & 3.75 & 662.4 & 5.03 & 664.8 \\
 & EFO & 1.50 & 636.6 & 2.03 & \underline{629.3} & 2.69 & 636.3 & 3.62 & 647.0 & 4.37 & \underline{643.5} \\
 & OPF & \underline{1.35} & 545.0 & 1.99 & 654.3 & 2.80 & 654.9 & 3.70 & 661.7 & 4.35 & 665.8 \\
 & OPR & 1.66 & 626.3 & 2.14 & 630.3 & 2.73 & \underline{633.0} & 3.60 & 646.0 & 4.51 & 656.7 \\
 & MOPP & 1.71 & 613.5 & 2.38 & 649.7 & 3.11 & 659.0 & 4.00 & 665.7 & 5.29 & 664.6 \\
 & OPST & 1.92 & 634.2 & 4.45 & 640.3 & 10.80 & 655.0 & 23.02 & 670.8 & 45.17 & 677.3 \\
 & SOPP & 1.36 & \underline{503.8} & 1.63 & 630.8 & 2.24 & 641.2 & 3.02 & \underline{640.2} & 3.72 & 648.0 \\
 & \CCG CT & \CCG \textbf{1.00} & \CCG \textbf{363.7} & \CCG \underline{1.43} & \CCG 652.1 & \CCG \underline{1.79} & \CCG 671.3 & \CCG \underline{2.09} & \CCG 680.6 & \CCG \underline{2.72} & \CCG 795.3 \\
\midrule
\multirow{10}{*}{MoteStrain} & Raw & 0.000 & 157.6 & 0.000 & 158.1 & 0.000 & 159.5 & 0.000 & 160.3 & 0.000 & 160.7 \\
\cmidrule{2-12}
 & PCA & -- & -- & \textbf{0.005} & \underline{163.6} & \textbf{0.007} & \underline{164.8} & \textbf{0.007} & \textbf{165.4} & \textbf{0.008} & \textbf{165.7} \\
 & OPP & 0.620 & \underline{157.9} & 1.40 & 385.3 & 2.07 & 636.1 & 2.81 & 637.6 & 4.36 & 643.3 \\
 & EFO & 0.581 & 178.0 & 0.979 & 411.7 & 1.60 & 633.2 & 2.40 & 617.5 & 3.74 & 639.7 \\
 & OPF & 0.540 & 159.0 & 1.00 & 338.8 & 1.45 & 617.2 & 2.43 & 639.2 & 3.51 & 639.2 \\
 & OPR & 0.641 & 183.5 & 1.10 & 424.2 & 1.75 & 616.9 & 2.61 & 619.3 & 4.15 & 636.9 \\
 & MOPP & 0.622 & 159.6 & 1.32 & 392.0 & 2.00 & 641.9 & 2.90 & 641.9 & 4.15 & 638.2 \\
 & OPST & \underline{0.411} & 158.5 & 0.770 & 253.5 & 1.36 & 515.8 & 2.24 & 625.9 & 3.41 & 629.8 \\
 & SOPP & 0.480 & 158.8 & 0.768 & 275.0 & 1.36 & 554.4 & 2.09 & 625.6 & 3.35 & 622.6 \\
 & \CCG CT & \CCG \textbf{0.281} & \CCG \textbf{157.7} & \CCG \underline{0.478} & \CCG \textbf{158.8} & \CCG \underline{0.577} & \CCG \textbf{159.2} & \CCG \underline{0.875} & \CCG \underline{230.0} & \CCG \underline{1.06} & \CCG \underline{299.9} \\
\midrule
\multirow{10}{*}{PigCVP} & Raw & 0.000 & 165.8 & 0.000 & 171.5 & 0.000 & 176.7 & 0.000 & 183.1 & 0.000 & 188.1 \\
\cmidrule{2-12}
 & PCA & \textbf{0.010} & \textbf{166.6} & \textbf{0.015} & \textbf{171.5} & \textbf{0.021} & \textbf{177.3} & \textbf{0.026} & \textbf{183.2} & \textbf{0.033} & \textbf{188.6} \\
 & OPP & 4.77 & 644.6 & 12.55 & 648.6 & 21.32 & 654.5 & 30.45 & 664.5 & 42.67 & 661.2 \\
 & EFO & 3.93 & 624.4 & 10.39 & 633.7 & 17.61 & 641.7 & 24.95 & 640.1 & 33.55 & 651.0 \\
 & OPF & 2.96 & 631.2 & 8.47 & 641.1 & 11.39 & 643.8 & 15.35 & 651.9 & 17.83 & 650.2 \\
 & OPR & 4.31 & 639.3 & 10.95 & 638.9 & 17.26 & 642.9 & 25.17 & 650.6 & 33.92 & 649.9 \\
 & MOPP & 4.38 & 638.8 & 10.18 & 644.7 & 17.24 & 647.9 & 26.58 & 653.9 & 37.60 & 656.8 \\
 & OPST & 2.88 & 633.7 & 6.96 & 637.1 & 10.65 & 637.9 & 15.03 & 639.5 & 18.69 & \underline{648.6} \\
 & SOPP & 2.91 & 625.8 & 7.76 & 628.4 & 11.41 & 636.0 & 14.69 & 652.5 & 18.11 & 662.9 \\
 & \CCG CT & \CCG \underline{0.635} & \CCG \underline{193.9} & \CCG \underline{0.982} & \CCG \underline{312.5} & \CCG \underline{1.63} & \CCG \underline{413.7} & \CCG \underline{1.66} & \CCG \underline{516.9} & \CCG \underline{1.92} & \CCG 679.3 \\
\bottomrule[1.2pt]
\end{tabular}
\end{table}
\FloatBarrier

\subsection{Experimental Results on Wafer}
Wafer contains 7,164 sequences of length 152 with two recorded classes.
Its clustering behavior differs sharply from the five main datasets
and can dominate macro-averaged comparisons.
At $K=30$, including Wafer increases the mean NMI gap between CT and EFO
from $0.0045$ to $0.1701$, with Wafer accounting for approximately
$97.8\%$ of the six-dataset gap.
We therefore report Wafer separately and provide its complete
quality and resource results in
Tables~\ref{CTPM:app:tab:wafer_quality}
and~\ref{CTPM:app:tab:wafer_resources}.

As shown in Table~\ref{CTPM:app:tab:wafer_quality},
CT reaches perfect clustering at $K=30$ and retains it at
$K=40$ and $K=50$.
Several frequency-based OP methods reach the same result only at $K=50$,
while SOPP reaches it at smaller budgets but does not retain it as $K$ grows.
This unusual behavior motivates the feature-level analysis presented
in the main text.

\begin{table}[th]\centering
\caption{
Clustering quality across feature budgets on Wafer.
Values are mean NMI, homogeneity, and ARI over five clustering seeds.
Bold and underline denote the best and second-best values excluding Raw.
}
\label{CTPM:app:tab:wafer_quality}
\small
\setlength{\tabcolsep}{2.0pt}
\begin{tabular}{lrrrrrrrrrrrrrrr}
\toprule[1.2pt]
\multirow{2}{*}[-2pt]{Method} & \multicolumn{3}{c}{$K=10$} & \multicolumn{3}{c}{$K=20$} & \multicolumn{3}{c}{$K=30$} & \multicolumn{3}{c}{$K=40$} & \multicolumn{3}{c}{$K=50$} \\
\cmidrule(lr){2-4}\cmidrule(lr){5-7}\cmidrule(lr){8-10}\cmidrule(lr){11-13}\cmidrule(lr){14-16}
 & NMI & $h$ & ARI & NMI & $h$ & ARI & NMI & $h$ & ARI & NMI & $h$ & ARI & NMI & $h$ & ARI \\
\midrule[1.2pt]
Raw & 0.0 & 0.0 & 0.1 & 0.0 & 0.0 & 0.1 & 0.0 & 0.0 & 0.1 & 0.0 & 0.0 & 0.1 & 0.0 & 0.0 & 0.1 \\
\midrule
PCA & 0.0 & 0.0 & 0.1 & 0.0 & 0.0 & 0.1 & 0.0 & 0.0 & 0.1 & 0.0 & 0.0 & 0.1 & 0.0 & 0.0 & 0.1 \\
OPP & 0.6 & 0.9 & 0.8 & 0.9 & 1.3 & -2.0 & 0.2 & 0.2 & -2.4 & 0.1 & 0.1 & -1.7 & \textbf{100.0} & \textbf{100.0} & \textbf{100.0} \\
EFO & 0.6 & 0.9 & 0.8 & 0.9 & 1.3 & -2.0 & 0.2 & 0.2 & -2.4 & 0.1 & 0.1 & -1.7 & \textbf{100.0} & \textbf{100.0} & \textbf{100.0} \\
OPF & 0.9 & 1.2 & 1.8 & 0.9 & 1.2 & -2.3 & 0.2 & 0.3 & -2.7 & 0.1 & 0.1 & -1.8 & \textbf{100.0} & \textbf{100.0} & \textbf{100.0} \\
OPR & 0.6 & 0.9 & 0.8 & 0.9 & 1.3 & -2.0 & 0.2 & 0.2 & -2.4 & 0.1 & 0.1 & -1.7 & \textbf{100.0} & \textbf{100.0} & \textbf{100.0} \\
MOPP & 0.6 & 0.9 & 0.8 & 0.9 & 1.3 & -2.0 & 0.2 & 0.2 & -2.4 & 0.1 & 0.1 & -1.7 & \textbf{100.0} & \textbf{100.0} & \textbf{100.0} \\
OPST & \underline{16.3} & \underline{23.4} & \underline{4.2} & \underline{17.0} & \underline{24.4} & \underline{5.4} & \underline{17.5} & \underline{25.0} & \underline{6.3} & \underline{19.9} & \underline{28.4} & \underline{10.5} & \underline{34.8} & \underline{45.8} & \underline{35.9} \\
SOPP & \textbf{100.0} & \textbf{100.0} & \textbf{100.0} & \textbf{100.0} & \textbf{100.0} & \textbf{100.0} & 0.8 & 1.1 & -2.3 & 0.1 & 0.1 & -2.0 & 0.0 & 0.0 & -0.6 \\
\CCG \textbf{CT} & \CCG 13.3 & \CCG 18.9 & \CCG -0.8 & \CCG 0.9 & \CCG 1.3 & \CCG -2.1 & \CCG \textbf{100.0} & \CCG \textbf{100.0} & \CCG \textbf{100.0} & \CCG \textbf{100.0} & \CCG \textbf{100.0} & \CCG \textbf{100.0} & \CCG \textbf{100.0} & \CCG \textbf{100.0} & \CCG \textbf{100.0} \\
\bottomrule[1.2pt]
\end{tabular}
\end{table}

\begin{table}[th]\centering
\caption{
Feature-construction time and peak memory across feature budgets on Wafer.
Time is averaged over five runs and memory is measured as peak process RSS.
Bold and underline denote the lowest and second-lowest values excluding Raw.
}
\label{CTPM:app:tab:wafer_resources}
\small
\setlength{\tabcolsep}{4.5pt}
\begin{tabular}{lrrrrrrrrrr}
\toprule[1.2pt]
\multirow{2}{*}[-2pt]{Method} & \multicolumn{2}{c}{$K=10$} & \multicolumn{2}{c}{$K=20$} & \multicolumn{2}{c}{$K=30$} & \multicolumn{2}{c}{$K=40$} & \multicolumn{2}{c}{$K=50$} \\
\cmidrule(lr){2-3}\cmidrule(lr){4-5}\cmidrule(lr){6-7}\cmidrule(lr){8-9}\cmidrule(lr){10-11}
 & Time. & Mem. & Time. & Mem. & Time. & Mem. & Time. & Mem. & Time. & Mem. \\
\midrule[1.2pt]
Raw & 0.000 & 214.8 & 0.000 & 214.8 & 0.000 & 214.8 & 0.000 & 214.8 & 0.000 & 214.8 \\
\midrule
PCA & \textbf{0.045} & \textbf{215.3} & \textbf{0.045} & \textbf{215.2} & \textbf{0.049} & \textbf{215.6} & \textbf{0.061} & \textbf{215.5} & \textbf{0.070} & \textbf{215.7} \\
OPP & 31.06 & 669.0 & 31.12 & 661.1 & 31.33 & 672.3 & 31.15 & 661.8 & 31.09 & 678.3 \\
EFO & 29.67 & 650.4 & 29.89 & 654.9 & 29.82 & 648.8 & 30.43 & \underline{645.9} & 30.06 & \underline{645.9} \\
OPF & 28.18 & 660.5 & 28.67 & 662.1 & 28.94 & 657.2 & 29.06 & 660.4 & 28.66 & 660.8 \\
OPR & 30.31 & 662.3 & 30.34 & 657.2 & 30.01 & 659.8 & 30.66 & 660.3 & 30.22 & 651.6 \\
MOPP & 18.72 & 675.1 & 19.48 & 663.2 & 19.84 & 665.4 & 19.15 & 660.7 & 19.04 & 659.7 \\
OPST & 33.74 & 655.8 & 33.98 & 655.0 & 33.22 & 655.3 & 33.38 & 656.0 & 33.06 & 655.9 \\
SOPP & 28.45 & \underline{646.4} & 28.71 & \underline{646.5} & 29.28 & \underline{646.6} & 29.50 & 646.4 & 29.42 & 647.7 \\
\CCG \textbf{CT} & \CCG \underline{2.01} & \CCG 665.8 & \CCG \underline{2.54} & \CCG 675.5 & \CCG \underline{2.93} & \CCG 774.7 & \CCG \underline{3.26} & \CCG 769.8 & \CCG \underline{3.57} & \CCG 765.9 \\
\bottomrule[1.2pt]
\end{tabular}
\end{table}

Table~\ref{CTPM:app:tab:wafer_resources} further illustrates
the time-memory trade-off of CT mining.
Across feature budgets, CT remains substantially faster than the
other pattern-mining methods, while its memory footprint becomes
larger at higher budgets.
This complements the main results by showing that the computational
benefit of shared Cartesian structure can come with additional
storage cost in some datasets.

The input-prefix results in
Tables~\ref{CTPM:app:tab:wafer_input_length_quality}
and~\ref{CTPM:app:tab:wafer_input_length_resources}
reveal substantial sensitivity to the amount of retained input.
As shown in Table~\ref{CTPM:app:tab:wafer_input_length_quality},
CT achieves nearly perfect clustering at the $20\%$ and $40\%$
prefixes and again at the full sequence, but not at the intermediate
$60\%$ and $80\%$ prefixes.
Several OP-based methods show similarly non-monotonic behavior.
Because the feature dictionary is reconstructed independently
for each prefix, these results indicate that clustering quality
depends strongly on which patterns enter the fixed top-$K$
representation, rather than simply on the amount of observed input.
Table~\ref{CTPM:app:tab:wafer_input_length_resources} provides the
corresponding resource measurements across the same prefixes.

\begin{table}[th]\centering
\caption{
Clustering quality across input prefixes on Wafer at $K=30$.
Values are mean NMI, homogeneity, and ARI over five clustering seeds.
Bold and underline denote the best and second-best values excluding Raw.
}
\label{CTPM:app:tab:wafer_input_length_quality}
\small
\setlength{\tabcolsep}{3.5pt}
\begin{tabular}{lrrrrrrrrrrrrrrr}
\toprule[1.2pt]
\multirow{2}{*}[-2pt]{Method} & \multicolumn{3}{c}{20\%} & \multicolumn{3}{c}{40\%} & \multicolumn{3}{c}{60\%} & \multicolumn{3}{c}{80\%} & \multicolumn{3}{c}{100\%} \\
\cmidrule(lr){2-4}\cmidrule(lr){5-7}\cmidrule(lr){8-10}\cmidrule(lr){11-13}\cmidrule(lr){14-16}
 & NMI & $h$ & ARI & NMI & $h$ & ARI & NMI & $h$ & ARI & NMI & $h$ & ARI & NMI & $h$ & ARI \\
\midrule[1.2pt]
Raw & 0.6 & 0.9 & 2.5 & 2.0 & 2.8 & 4.5 & 0.0 & 0.0 & 0.1 & 0.0 & 0.1 & 0.6 & 0.0 & 0.0 & 0.1 \\
\midrule
PCA & 0.6 & 0.9 & 2.5 & 2.0 & 2.8 & \underline{4.5} & 0.0 & 0.0 & \underline{0.1} & 0.0 & 0.1 & \textbf{0.6} & 0.0 & 0.0 & 0.1 \\
OPP & \underline{97.8} & \underline{97.4} & \underline{99.3} & 3.8 & 5.2 & -4.8 & 1.8 & 1.8 & -7.0 & 4.5 & 6.0 & -6.0 & 0.2 & 0.2 & -2.4 \\
EFO & \underline{97.8} & \underline{97.4} & \underline{99.3} & 3.8 & 5.2 & -4.8 & 1.8 & 1.8 & -7.0 & 4.5 & 6.0 & -6.0 & 0.2 & 0.2 & -2.4 \\
OPF & 94.1 & 93.0 & 97.6 & 5.1 & 7.0 & -4.7 & 1.9 & 1.9 & -7.1 & 2.8 & 3.8 & \underline{-5.2} & 0.2 & 0.3 & -2.7 \\
OPR & \underline{97.8} & \underline{97.4} & \underline{99.3} & 3.8 & 5.2 & -4.8 & 1.8 & 1.8 & -7.0 & 4.5 & 6.0 & -6.0 & 0.2 & 0.2 & -2.4 \\
MOPP & \underline{97.8} & \underline{97.4} & \underline{99.3} & 3.8 & 5.2 & -4.8 & 1.8 & 1.8 & -7.0 & 4.5 & 6.0 & -6.0 & 0.2 & 0.2 & -2.4 \\
OPST & 8.8 & 11.7 & -7.3 & \underline{10.1} & \underline{13.9} & -5.6 & \underline{10.1} & \underline{13.9} & -5.6 & \textbf{14.2} & \textbf{20.2} & \textbf{0.6} & \underline{17.5} & \underline{25.0} & \underline{6.3} \\
SOPP & 89.6 & 87.5 & 95.2 & 3.5 & 4.7 & -5.1 & \textbf{21.5} & \textbf{21.6} & \textbf{14.2} & 3.7 & 4.9 & -5.7 & 0.8 & 1.1 & -2.3 \\
\CCG \textbf{CT} & \CCG \textbf{99.6} & \CCG \textbf{99.6} & \CCG \textbf{99.9} & \CCG \textbf{99.6} & \CCG \textbf{99.6} & \CCG \textbf{99.9} & \CCG 2.0 & \CCG 2.1 & \CCG -7.3 & \CCG \underline{4.7} & \CCG \underline{6.3} & \CCG -6.0 & \CCG \textbf{100.0} & \CCG \textbf{100.0} & \CCG \textbf{100.0} \\
\bottomrule[1.2pt]
\end{tabular}
\end{table}

\begin{table}[th]\centering
\caption{
Feature-construction time and peak memory across input prefixes
on Wafer at $K=30$.
Time is averaged over five runs and memory is measured as peak process RSS.
Bold and underline denote the lowest and second-lowest values excluding Raw.
}
\label{CTPM:app:tab:wafer_input_length_resources}
\small
\setlength{\tabcolsep}{4.5pt}
\begin{tabular}{lrrrrrrrrrr}
\toprule[1.2pt]
\multirow{2}{*}[-2pt]{Method} & \multicolumn{2}{c}{20\%} & \multicolumn{2}{c}{40\%} & \multicolumn{2}{c}{60\%} & \multicolumn{2}{c}{80\%} & \multicolumn{2}{c}{100\%} \\
\cmidrule(lr){2-3}\cmidrule(lr){4-5}\cmidrule(lr){6-7}\cmidrule(lr){8-9}\cmidrule(lr){10-11}
 & Time. & Mem. & Time. & Mem. & Time. & Mem. & Time. & Mem. & Time. & Mem. \\
\midrule[1.2pt]
Raw & 0.000 & 174.7 & 0.000 & 185.7 & 0.000 & 194.3 & 0.000 & 204.0 & 0.000 & 215.4 \\
\midrule
PCA & \textbf{0.018} & \textbf{174.0} & \textbf{0.029} & \textbf{185.6} & \textbf{0.037} & \textbf{194.2} & \textbf{0.044} & \textbf{204.7} & \textbf{0.053} & \textbf{215.0} \\
OPP & 1.65 & 627.7 & 3.68 & 661.5 & 8.27 & 650.8 & 10.64 & 651.3 & 30.74 & 671.8 \\
EFO & 1.40 & 631.6 & 2.94 & 632.8 & 7.26 & 635.9 & 10.45 & 642.9 & 30.14 & 651.7 \\
OPF & 1.49 & 559.7 & 2.86 & 648.6 & 6.86 & 645.7 & 9.58 & 650.4 & 28.41 & 658.5 \\
OPR & 1.52 & 647.3 & 3.19 & 646.6 & 7.50 & 651.6 & 10.85 & 655.0 & 30.50 & 659.1 \\
MOPP & 1.83 & 629.9 & 3.50 & 651.7 & 8.45 & 654.1 & 10.75 & 652.9 & 19.39 & 660.7 \\
OPST & 1.81 & 626.5 & 3.77 & 637.6 & 11.50 & 635.9 & 19.11 & 641.8 & 33.22 & 655.5 \\
SOPP & 1.31 & 475.4 & 2.37 & 630.6 & 6.60 & \underline{631.3} & 9.72 & \underline{637.9} & 28.90 & \underline{646.7} \\
\CCG \textbf{CT} & \CCG \underline{0.869} & \CCG \underline{293.5} & \CCG \underline{1.48} & \CCG \underline{546.6} & \CCG \underline{2.35} & \CCG 670.6 & \CCG \underline{2.62} & \CCG 677.2 & \CCG \underline{2.80} & \CCG 768.0 \\
\bottomrule[1.2pt]
\end{tabular}
\end{table}

\paragraph{Detailed feature statistics.}
Table~\ref{CTPM:app:tab:wafer_features} provides the numerical support for the
feature-level examples discussed in the main text.
The class-wise count ranges are disjoint for all four representative
CT features.
These statistics use unstandardized occurrence counts and are computed
only after evaluation, without affecting feature selection or clustering.

\begin{table}[th]
    \centering
    \small
    \caption{
    Detailed statistics for representative Wafer CT features from the
    saved $K=30$ dictionary.
    Class means and ranges use unstandardized occurrence counts.
    }
    \label{CTPM:app:tab:wafer_features}
    \begin{tabular}{lrrrrrr}
    \toprule[1.2pt]
    Rank & CT code & Total support & Mean $-1$ & Range $-1$ & Mean $+1$ & Range $+1$ \\
    \midrule[1.2pt]
    8 & $(0,0,1)$ & 57,857 & 50.71 & [28, 64] & 3.00 & [0, 17] \\
    25 & $(0,1,0,1)$ & 16,507 & 20.00 & [7, 32] & 0.20 & [0, 4] \\
    29 & $(0,0,1,0)$ & 15,482 & 18.05 & [8, 30] & 0.27 & [0, 6] \\
    30 & $(0,0,1,2)$ & 15,227 & 18.12 & [7, 29] & 0.22 & [0, 6] \\
    \bottomrule[1.2pt]
    \end{tabular}
\end{table}

\subsection{Full Results on Comparison with Naive Implementations}

Let $n$ denote the length of a sequence, $\ell$ an individual
pattern length, and $L$ the maximum evaluated length.
A direct pairwise implementation compares all pairs among the
$n-\ell+1$ windows of length $\ell$, with each exact CT comparison
requiring $O(\ell)$ time.
Its total worst-case cost is therefore
\[
 T_{\mathrm{pair}}(n,L)
 =O\!\left(\sum_{\ell=2}^{L}(n-\ell+1)^2\ell\right)
 =O(n^2L^2),
\]
which becomes $O(n^4)$ when $L=\Theta(n)$.
Because this direct implementation is prohibitively expensive,
we use CT-Hash as the practical naive baseline in our experiments.

CT-Hash independently constructs a Cartesian-tree encoding
for every window using a monotone stack and groups identical
encodings with a hash table.
Constructing and hashing a length-$\ell$ encoding requires
$O(\ell)$ time, giving the expected cumulative cost
\[
 \mathbb{E}[T_{\mathrm{hash}}(n,L)]
 =O\!\left(\sum_{\ell=2}^{L}(n-\ell+1)\ell\right)
 =O(nL^2).
\]
Thus, CT-Hash is linear in $n$ for fixed $L$ and cubic when
$L=\Theta(n)$, while our CT implementation reuses shared
Cartesian structure across pattern lengths.

We evaluate $L\in\{8,16,32,64,128,256\}$ whenever $L\leq n$.
Both methods process every length from $2$ to $L$ with early stopping
disabled and use identical pattern ranking with $K=30$.
We verify that their selected patterns, support values, and feature
matrices match in every paired condition.
Reported time covers the complete feature-construction pipeline
and excludes clustering.

The full results show that the runtime advantage is not uniform
across all sequence lengths and datasets.
CT is faster than CT-Hash at every feasible $L$ on MoteStrain,
while ElectricDevices consistently favors CT-Hash.
For the longer Car, Beef, and PigCVP sequences, however,
the runtime gap generally widens as $L$ increases.
The memory results also confirm the time-space trade-off discussed
in the main text, with the most pronounced growth appearing on PigCVP.

\begin{table}[th]
\centering
\caption{
Comparison with the naive CT-Hash implementation as the maximum
pattern length $L$ increases.
We report five-run mean feature-construction time and
single-run peak process RSS.
}
\label{CTPM:app:tab:comparison_naive}
\small
\setlength{\tabcolsep}{4.5pt}
\begin{tabular}{lllrrrrrr}
\toprule[1.2pt]
Dataset & Algorithm & Measure & $L=8$ & $L=16$ & $L=32$ & $L=64$ & $L=128$ & $L=256$ \\
\midrule[1.2pt]
\multirow{4}{*}[-3pt]{\shortstack[c]{Car\\($n=577$)}} & \multirow{2}{*}{CT} & Time. (s) & 0.357 & 0.508 & 0.818 & 1.370 & 2.333 & 3.656 \\
 &  & Mem. (MiB) & 127.2 & 152.0 & 251.4 & 490.3 & 692.5 & 901.4 \\
\cmidrule(lr){2-9}
 & \multirow{2}{*}{CT-Hash} & Time. (s) & 0.467 & 0.830 & 1.700 & 3.929 & 10.162 & 25.691 \\
 &  & Mem. (MiB) & 225.0 & 651.3 & 652.7 & 676.7 & 708.8 & 723.4 \\
\midrule
\multirow{4}{*}[-3pt]{\shortstack[c]{Beef\\($n=470$)}} & \multirow{2}{*}{CT} & Time. (s) & 0.283 & 0.353 & 0.563 & 0.974 & 1.456 & 1.819 \\
 &  & Mem. (MiB) & 127.2 & 127.2 & 166.3 & 322.4 & 578.0 & 681.7 \\
\cmidrule(lr){2-9}
 & \multirow{2}{*}{CT-Hash} & Time. (s) & 0.303 & 0.493 & 1.027 & 1.926 & 4.151 & 9.163 \\
 &  & Mem. (MiB) & 132.8 & 321.2 & 640.6 & 658.3 & 661.3 & 650.3 \\
\midrule
\multirow{4}{*}[-3pt]{\shortstack[c]{Electric\\Devices\\($n=96$)}} & \multirow{2}{*}{CT} & Time. (s) & 2.933 & 3.521 & 4.772 & 7.650 & -- & -- \\
 &  & Mem. (MiB) & 731.2 & 741.6 & 839.9 & 1,525.0 & -- & -- \\
\cmidrule(lr){2-9}
 & \multirow{2}{*}{CT-Hash} & Time. (s) & 2.412 & 3.153 & 4.127 & 5.399 & -- & -- \\
 &  & Mem. (MiB) & 675.1 & 676.4 & 691.1 & 693.0 & -- & -- \\
\midrule
\multirow{4}{*}[-3pt]{\shortstack[c]{MoteStrain\\($n=84$)}} & \multirow{2}{*}{CT} & Time. (s) & 0.546 & 0.818 & 1.337 & 1.751 & -- & -- \\
 &  & Mem. (MiB) & 161.3 & 230.9 & 417.9 & 668.2 & -- & -- \\
\cmidrule(lr){2-9}
 & \multirow{2}{*}{CT-Hash} & Time. (s) & 0.626 & 1.208 & 2.241 & 3.704 & -- & -- \\
 &  & Mem. (MiB) & 332.3 & 642.4 & 671.0 & 675.4 & -- & -- \\
\midrule
\multirow{4}{*}[-3pt]{\shortstack[c]{PigCVP\\($n=2000$)}} & \multirow{2}{*}{CT} & Time. (s) & 1.697 & 2.712 & 4.369 & 6.692 & 10.555 & 17.398 \\
 &  & Mem. (MiB) & 551.1 & 667.8 & 993.9 & 1,436.2 & 3,761.1 & 6,841.0 \\
\cmidrule(lr){2-9}
 & \multirow{2}{*}{CT-Hash} & Time. (s) & 1.669 & 4.766 & 10.206 & 18.916 & 33.675 & 62.629 \\
 &  & Mem. (MiB) & 653.1 & 730.4 & 739.7 & 755.6 & 752.3 & 753.7 \\
\bottomrule[1.2pt]
\end{tabular}
\end{table}

\subsection{Threshold-based Feature Selection}

The main experiments control representation size using a fixed feature
budget.
We additionally evaluate the native threshold-based setting, where all
patterns satisfying the mining criteria are retained.
We vary the minimum support over
$\mathrm{minsup}\in\{5,10,25\}$.
For OPR and SOPP, which introduce additional confidence and stability
criteria, we report the setting with the highest mean NMI for each
dataset and support threshold.
Table~\ref{CTPM:app:tab:threshold} summarizes the resulting feature counts,
clustering quality, and feature-construction time.
These results are intended as a sensitivity analysis rather than the
primary controlled comparison.

For a fixed pattern length and the same set of valid windows,
a CT class combines the disjoint OP classes that map to the same
Cartesian-tree structure.
Its support can therefore be written as
\[
 \mathrm{supp}_{\mathrm{CT}}(c)
 =
 \sum_{p:\,\mathrm{CT}(p)=c}
 \mathrm{supp}_{\mathrm{OP}}(p).
\]
Consequently, several OP classes that individually fall below the
support threshold may form a CT class that exceeds it.
A coarser equivalence relation therefore does not necessarily produce
fewer threshold-qualified features.

As shown in Table~\ref{CTPM:app:tab:threshold}, the number of retained
features differs substantially across equivalence relations.
CT often returns more features than ordinary frequent-OP mining,
particularly at lower support thresholds, which is consistent with
support being aggregated across multiple OP classes.
This behavior is especially visible on Car, Beef, and PigCVP.

Clustering quality also depends strongly on the support threshold.
CT is among the strongest pattern-based representations at higher
thresholds on ElectricDevices, MoteStrain, and PigCVP, while no single
threshold is uniformly favorable across all datasets.
Because the methods retain substantially different numbers of features,
the runtimes in Table~\ref{CTPM:app:tab:threshold} reflect their complete
threshold-based workloads rather than a comparison at equal
representation size.
These results complement the fixed-budget experiments and also motivate
using a common feature budget for the primary comparison.

\begin{table}[th]\centering
\caption{
Threshold-based feature selection on the five main datasets.
All qualifying features are retained at each minimum support.
Dim denotes the resulting feature count and Time denotes mean
feature-construction time.
OPR and SOPP use the best mean-NMI setting from their parameter grids.
Bold and underline denote the best and second-best quality among
pattern-mining methods.
}
\label{CTPM:app:tab:threshold}
\small
\setlength{\tabcolsep}{1.70pt}
\begin{tabular}{llrrrrrrrrrrrrrrr}
\toprule[1.2pt]
\multirow{2}{*}[-2pt]{Dataset} & \multirow{2}{*}[-2pt]{Method} & \multicolumn{5}{c}{minsup = 5} & \multicolumn{5}{c}{minsup = 10} & \multicolumn{5}{c}{minsup = 25} \\
\cmidrule(lr){3-7}\cmidrule(lr){8-12}\cmidrule(lr){13-17}
 &  & Dim & NMI & $h$ & ARI & Time. & Dim & NMI & $h$ & ARI & Time. & Dim & NMI & $h$ & ARI & Time. \\
\midrule[1.2pt]
\multirow{10}{*}{Car} & Raw & 577 & 22.8 & 21.2 & 12.6 & 0.00 & 577 & 22.8 & 21.2 & 12.6 & 0.00 & 577 & 22.8 & 21.2 & 12.6 & 0.00 \\
 & PCA & 30 & 22.8 & 21.2 & 12.6 & 0.01 & 30 & 22.8 & 21.2 & 12.6 & 0.01 & 30 & 22.8 & 21.2 & 12.6 & 0.01 \\
\cmidrule{2-17}
 & OPP & 596 & 27.1 & 25.6 & 20.1 & 15.99 & 293 & \textbf{28.4} & \textbf{26.3} & \textbf{22.7} & 15.14 & 219 & \textbf{25.5} & \underline{23.7} & \underline{19.9} & 11.59 \\
 & EFO & 596 & 27.1 & 25.6 & 20.1 & 14.78 & 293 & \textbf{28.4} & \textbf{26.3} & \textbf{22.7} & 14.03 & 219 & \textbf{25.5} & \underline{23.7} & \underline{19.9} & 10.88 \\
 & OPF & 334 & 23.5 & 21.4 & 17.5 & 14.74 & 253 & 23.9 & 21.8 & 18.2 & 13.46 & 208 & 23.0 & 21.9 & 17.6 & 9.82 \\
 & OPR & 492 & \textbf{30.2} & \textbf{28.3} & \textbf{24.2} & 15.06 & 270 & 26.4 & 24.3 & \underline{20.3} & 14.27 & 218 & \textbf{25.5} & \textbf{23.8} & \textbf{20.1} & 11.11 \\
 & MOPP & 348 & 27.6 & 26.0 & 20.3 & 13.24 & 165 & 26.5 & 24.4 & 19.3 & 12.09 & 113 & 23.5 & 22.0 & 16.9 & 9.74 \\
 & OPST & 348 & 27.6 & 26.0 & 20.3 & 12.34 & 165 & 26.5 & 24.4 & 19.3 & 11.12 & 113 & 23.5 & 22.0 & 16.9 & 8.72 \\
 & SOPP & 336 & \underline{29.8} & \underline{27.1} & \underline{22.7} & 14.73 & 241 & 25.6 & 23.6 & 19.4 & 13.91 & 163 & 24.1 & 22.4 & 18.1 & 10.81 \\
 & \CCG CT & \CCG 735 & \CCG 23.2 & \CCG 21.7 & \CCG 16.0 & \CCG 10.16 & \CCG 324 & \CCG \underline{26.9} & \CCG \underline{24.8} & \CCG 20.0 & \CCG 9.61 & \CCG 223 & \CCG \underline{24.9} & \CCG 23.2 & \CCG 19.3 & \CCG 7.66 \\
\midrule
\multirow{10}{*}{Beef} & Raw & 470 & 30.3 & 27.3 & 8.4 & 0.00 & 470 & 30.3 & 27.3 & 8.4 & 0.00 & 470 & 30.3 & 27.3 & 8.4 & 0.00 \\
 & PCA & 30 & 30.3 & 27.3 & 8.4 & 0.00 & 30 & 30.3 & 27.3 & 8.4 & 0.00 & 30 & 30.3 & 27.3 & 8.4 & 0.00 \\
\cmidrule{2-17}
 & OPP & 366 & 16.6 & 16.4 & 2.6 & 1.97 & 136 & 20.5 & 19.1 & 5.6 & 1.65 & 79 & 23.7 & 22.2 & 6.0 & 1.04 \\
 & EFO & 366 & 16.6 & 16.4 & 2.6 & 1.92 & 136 & 20.5 & 19.1 & 5.6 & 1.58 & 79 & 23.7 & 22.2 & 6.0 & 0.988 \\
 & OPF & 173 & \textbf{27.0} & \textbf{24.8} & \textbf{10.3} & 1.85 & 105 & \textbf{34.0} & \textbf{29.9} & \textbf{13.3} & 1.53 & 61 & \textbf{31.0} & \textbf{30.0} & \textbf{14.2} & 0.758 \\
 & OPR & 324 & 17.5 & 16.9 & 4.3 & 2.02 & 119 & \underline{22.9} & \underline{21.1} & 6.4 & 1.66 & 75 & 23.8 & 22.4 & 6.1 & 1.05 \\
 & MOPP & 240 & 18.3 & 17.9 & 4.7 & 1.74 & 89 & 17.2 & 16.9 & 4.1 & 1.35 & 40 & 21.3 & 19.9 & 5.4 & 0.859 \\
 & OPST & 240 & 18.3 & 17.9 & 4.7 & 1.47 & 89 & 17.2 & 16.9 & 4.1 & 1.25 & 40 & 21.3 & 19.9 & 5.4 & 0.787 \\
 & SOPP & 325 & 19.3 & 18.4 & 5.8 & 1.70 & 75 & 21.5 & 18.8 & \underline{7.1} & 1.34 & 27 & 23.7 & 21.8 & \underline{6.5} & 0.638 \\
 & \CCG CT & \CCG 520 & \CCG \underline{23.0} & \CCG \underline{22.0} & \CCG \underline{9.2} & \CCG 3.61 & \CCG 169 & \CCG 19.4 & \CCG 18.5 & \CCG 4.4 & \CCG 1.04 & \CCG 81 & \CCG \underline{24.6} & \CCG \underline{23.1} & \CCG \underline{6.5} & \CCG 0.959 \\
\midrule
\multirow{10}{*}{\begin{tabular}[t]{@{}l@{}}Electric\\Devices\end{tabular}} & Raw & 96 & 19.9 & 19.3 & 16.8 & 0.00 & 96 & 19.9 & 19.3 & 16.8 & 0.00 & 96 & 19.9 & 19.3 & 16.8 & 0.00 \\
 & PCA & 30 & 20.1 & 19.4 & 17.9 & 0.09 & 30 & 20.1 & 19.4 & 17.9 & 0.09 & 30 & 20.1 & 19.4 & 17.9 & 0.09 \\
\cmidrule{2-17}
 & OPP & 250 & 11.7 & 5.7 & -0.4 & 3.75 & 72 & 18.7 & 12.7 & 1.6 & 3.30 & 8 & \underline{33.6} & \underline{31.4} & 12.1 & 1.74 \\
 & EFO & 250 & 11.7 & 5.7 & -0.4 & 3.51 & 72 & 18.7 & 12.7 & 1.6 & 2.92 & 8 & \underline{33.6} & \underline{31.4} & 12.1 & 1.60 \\
 & OPF & 101 & 17.8 & \underline{11.6} & 1.3 & 3.52 & 26 & \textbf{26.4} & \textbf{22.1} & \underline{7.9} & 2.72 & 2 & 24.9 & 25.2 & 11.3 & 1.57 \\
 & OPR & 249 & \underline{18.0} & 11.5 & \underline{3.6} & 3.79 & 65 & 20.7 & 14.4 & 2.8 & 3.02 & 7 & \textbf{34.0} & \textbf{31.7} & \underline{12.8} & 1.60 \\
 & MOPP & 220 & 15.6 & 9.2 & 0.9 & 3.77 & 71 & 18.7 & 12.6 & 1.6 & 3.17 & 8 & \underline{33.6} & \underline{31.4} & 12.1 & 1.93 \\
 & OPST & 2,260 & \textbf{30.3} & \textbf{24.6} & \textbf{9.1} & 45.58 & 759 & \textbf{26.4} & \underline{21.9} & \textbf{12.7} & 43.92 & 175 & 31.0 & 28.5 & \textbf{26.5} & 36.95 \\
 & SOPP & 249 & 12.3 & 6.4 & 0.5 & 2.98 & 42 & \underline{25.0} & 20.7 & 6.9 & 1.86 & 8 & \underline{33.6} & \underline{31.4} & 12.1 & 1.49 \\
 & \CCG CT & \CCG 218 & \CCG 11.7 & \CCG 6.8 & \CCG 0.1 & \CCG 3.92 & \CCG 69 & \CCG 20.8 & \CCG 14.3 & \CCG 2.7 & \CCG 3.31 & \CCG 7 & \CCG \textbf{34.0} & \CCG \textbf{31.7} & \CCG \underline{12.8} & \CCG 2.42 \\
\midrule
\multirow{10}{*}{MoteStrain} & Raw & 84 & 40.9 & 40.9 & 51.0 & 0.00 & 84 & 40.9 & 40.9 & 51.0 & 0.00 & 84 & 40.9 & 40.9 & 51.0 & 0.00 \\
 & PCA & 30 & 40.7 & 40.8 & 50.9 & 0.01 & 30 & 40.7 & 40.8 & 50.9 & 0.01 & 30 & 40.7 & 40.8 & 50.9 & 0.01 \\
\cmidrule{2-17}
 & OPP & 155 & 20.4 & 16.3 & 5.6 & 3.66 & 83 & 19.6 & 15.4 & 4.9 & 3.09 & 48 & \textbf{43.3} & \underline{42.3} & \underline{40.9} & 1.67 \\
 & EFO & 155 & 20.4 & 16.3 & 5.6 & 3.23 & 83 & 19.6 & 15.4 & 4.9 & 2.69 & 48 & \textbf{43.3} & \underline{42.3} & \underline{40.9} & 1.50 \\
 & OPF & 97 & \textbf{21.0} & \textbf{16.9} & \textbf{6.2} & 3.08 & 71 & \textbf{22.4} & \textbf{18.5} & \textbf{7.7} & 2.52 & 31 & \underline{43.1} & \textbf{42.7} & \textbf{46.3} & 1.27 \\
 & OPR & 155 & 20.4 & 16.3 & 5.6 & 3.31 & 83 & 19.6 & 15.4 & 4.9 & 2.78 & 48 & \textbf{43.3} & \underline{42.3} & \underline{40.9} & 1.62 \\
 & MOPP & 144 & \underline{20.5} & \underline{16.4} & \underline{5.7} & 3.62 & 79 & 19.6 & 15.4 & 4.9 & 3.10 & 45 & 42.9 & 41.8 & 40.3 & 1.61 \\
 & OPST & 200 & \underline{20.5} & \underline{16.4} & \underline{5.7} & 3.12 & 107 & 20.3 & 16.2 & \underline{5.6} & 2.55 & 58 & \textbf{43.3} & \underline{42.3} & \underline{40.9} & 1.58 \\
 & SOPP & 155 & 20.4 & 16.3 & 5.6 & 2.84 & 83 & 19.6 & 15.4 & 4.9 & 2.37 & 48 & \textbf{43.3} & \underline{42.3} & \underline{40.9} & 1.45 \\
 & \CCG CT & \CCG 161 & \CCG 20.4 & \CCG 16.3 & \CCG 5.6 & \CCG 2.21 & \CCG 88 & \CCG \underline{20.4} & \CCG \underline{16.3} & \CCG \underline{5.6} & \CCG 1.80 & \CCG 48 & \CCG \textbf{43.3} & \CCG \underline{42.3} & \CCG \underline{40.9} & \CCG 1.45 \\
\midrule
\multirow{10}{*}{PigCVP} & Raw & 2,000 & 55.9 & 52.6 & 5.6 & 0.00 & 2,000 & 55.9 & 52.6 & 5.6 & 0.00 & 2,000 & 55.9 & 52.6 & 5.6 & 0.00 \\
 & PCA & 30 & 57.6 & 55.0 & 6.0 & 0.03 & 30 & 57.6 & 55.0 & 6.0 & 0.03 & 30 & 57.6 & 55.0 & 6.0 & 0.03 \\
\cmidrule{2-17}
 & OPP & 2,708 & 61.6 & 57.4 & 11.5 & 19.21 & 811 & 66.2 & 62.9 & 17.1 & 14.21 & 234 & \underline{68.7} & \underline{66.0} & \underline{20.8} & 6.91 \\
 & EFO & 2,708 & 61.6 & 57.4 & 11.5 & 16.97 & 811 & 66.2 & 62.9 & 17.1 & 13.40 & 234 & \underline{68.7} & \underline{66.0} & \underline{20.8} & 6.67 \\
 & OPF & 1,095 & \textbf{64.2} & \textbf{60.9} & \textbf{14.1} & 13.58 & 436 & \underline{66.8} & \underline{63.9} & \underline{17.6} & 8.36 & 129 & 66.1 & 63.7 & 16.2 & 5.14 \\
 & OPR & 2,636 & 61.4 & 57.1 & 11.4 & 17.22 & 534 & 66.7 & 63.4 & 17.1 & 13.54 & 216 & 67.6 & 65.0 & 18.7 & 6.63 \\
 & MOPP & 2,294 & 61.7 & 57.6 & 11.2 & 14.37 & 749 & 65.7 & 62.6 & 15.9 & 10.07 & 215 & 68.5 & 65.9 & 19.7 & 7.25 \\
 & OPST & 2,430 & 62.2 & 58.7 & \underline{12.0} & 15.76 & 770 & 66.1 & 63.3 & 16.6 & 12.68 & 216 & 67.4 & 64.9 & 18.5 & 8.52 \\
 & SOPP & 1,512 & \underline{62.8} & \underline{59.3} & \underline{12.0} & 15.32 & 301 & 66.7 & \textbf{64.0} & 17.2 & 12.70 & 209 & 68.0 & 65.6 & 19.5 & 5.73 \\
 & \CCG CT & \CCG 3,858 & \CCG 62.1 & \CCG 57.1 & \CCG 11.7 & \CCG 16.25 & \CCG 1,266 & \CCG \textbf{67.0} & \CCG 63.4 & \CCG \textbf{18.3} & \CCG 15.44 & \CCG 324 & \CCG \textbf{70.3} & \CCG \textbf{67.3} & \CCG \textbf{23.1} & \CCG 6.76 \\
\bottomrule[1.2pt]
\end{tabular}
\end{table}

\end{document}